\documentclass[11pt]{article}
\usepackage[margin=1in]{geometry}

\usepackage{booktabs} 

\usepackage[utf8]{inputenc}
\usepackage{authblk}
\usepackage{amsmath, amssymb, amsthm, thmtools, amsfonts, bm, bbm, thm-restate}
\usepackage[ruled,noend,linesnumbered]{algorithm2e}
\SetArgSty{textup}

\usepackage{asymptote}
\usepackage{graphicx}
\usepackage{todonotes}
\usepackage{multirow}
\usepackage{comment}
\usepackage{dsfont}
\usepackage{bigstrut}
\usepackage{caption}
\usepackage{subcaption}

\usepackage{xcolor}
\definecolor{BrickRed}{rgb}{0.8,0.25,0.33}
\usepackage[pagebackref,colorlinks,citecolor=blue,linkcolor=BrickRed]{hyperref}
\usepackage[framemethod=tikz]{mdframed}
\usepackage{nicefrac}

\usepackage{pifont}

\usepackage{cleveref}

\theoremstyle{plain}
\newtheorem{thm}{Theorem}[section]
\newtheorem{cor}[thm]{Corollary}

\newtheorem{fact}[thm]{Fact}
\newtheorem{lem}[thm]{Lemma}

\newtheorem{Def}[thm]{Definition}

\newtheorem{rem}[thm]{Remark}

\crefname{thm}{Theorem}{theorems}
\crefname{cla}{Claim}{claims}
\crefname{lem}{Lemma}{lemmas}
\crefname{fact}{Fact}{facts}

\newcommand{\e}{\texttt{e}}
\newcommand{\E}{\mathbb{E}}

\newcommand{\eps}{\varepsilon}
\newcommand{\calA}{\mathcal{A}}

\newcommand{\calF}{\mathcal{F}}

\newcommand{\calP}{\mathcal{P}}

\newcommand{\corrgap}{0.45}
\newcommand{\lift}{0.017}
\newcommand{\bipcorrgap}{0.456}
\newcommand{\patiencratio}{0.395}
\newcommand{\bippatienceratio}{0.426}

\newenvironment{wrapper}[1]
{
	\smallskip
	\begin{center}
		\begin{minipage}{\linewidth}
			\begin{mdframed}[hidealllines=true, backgroundcolor=gray!20, leftmargin=0cm,innerleftmargin=0.35cm,innerrightmargin=0.35cm,innertopmargin=0.375cm,innerbottommargin=0.375cm,roundcorner=10pt]
				#1}
			{\end{mdframed}
		\end{minipage}
	\end{center}
	\smallskip
}

\begin{document}

\title{Improved Online Contention Resolution for Matchings and Applications to the Gig Economy}

\author[1]{Tristan Pollner}
\author[1]{Mohammad Roghani}
\author[1]{Amin Saberi}
\author[1]{David Wajc}

\affil[1]{Stanford University}

\date{}

\maketitle
\pagenumbering{gobble}
\begin{abstract}

Motivated by applications in the gig economy, we study approximation algorithms for a \emph{sequential pricing problem}. The input is a bipartite graph $G=(I,J,E)$ between individuals $I$ and jobs $J$. The platform has a value of $v_j$ for matching job $j$ to an individual worker.
In order to find a matching, the platform can consider the edges  $(i j) \in E$ in any order and make a one-time take-it-or-leave-it offer of a price $\pi_{ij} = w$ of its choosing to $i$ for completing $j$. 
The worker accepts the offer with a known probability $ p_{ijw} $; in this case the job and the worker are irrevocably matched. 
What is the best way to make offers to maximize revenue and/or social welfare?

    The optimal algorithm is known to be NP-hard to compute (even if there is only a single job). With this in mind, we design efficient approximations to the optimal policy via a new Random-Order Online Contention Resolution Scheme (RO-OCRS) for matching. Our main result is a 0.456-balanced RO-OCRS in bipartite graphs and a 0.45-balanced RO-OCRS in general graphs. These algorithms improve on the recent bound of $\frac{1}{2}(1-\e^{-2})\approx 0.432$ of \cite{brubach2021improved}, and improve on the best known lower bounds for the \emph{correlation gap} of matching, despite applying to a significantly more restrictive setting.  As a consequence of our OCRS results, we obtain a $0.456$-approximate algorithm for the sequential pricing problem. We further extend our results to settings where workers can only be contacted a limited number of times, and show how to achieve improved results for this problem, via improved algorithms for the well-studied stochastic probing problem.
    
\end{abstract}


\newpage

\pagenumbering{arabic}
\section{Introduction}

Online gig platforms form a large and growing section of the  economy. Gallup estimates that 36\% of U.S. workers participate in the ``gig economy'' through either their primary or secondary jobs. This includes  multiple types of alternative work arrangements from temp workers and contract nurses to millions of online platform workers (for platforms like DoorDash, Uber or TaskRabbit) providing rides, shopping, and delivery services. 

One of the salient features of these app-based online platforms is that they pair workers with tasks in real-time. As a result of that, and because of interface limitations of mobile apps, pricing the tasks and matching them to workers cannot be done in one shot. Instead, platforms usually find matches for tasks by sending a take-it-or-leave-it offer through the app to the workers, detailing the task description and how much they will be paid, with a short window to decide. If one worker declines the offer, the job will be sent to another, possibly with a different price. Motivated by this, we consider the following problem abstraction.

\paragraph{The Sequential Pricing Problem.} The input is a bipartite graph $G=(I,J,E)$ with vertices on one side corresponding to the available \emph{individuals (workers)} $I$ and \emph{jobs} $J$. The platform can consider the edges in $E$ in any order; when considering edge $(ij)$, it makes a one-time take-it-or-leave-it offer of price $\pi_{ij}=w$ to worker $i$ for completing job $j$. The worker accepts with a (known) probability $p_{ijw} = p_{ew}$. In this case, the worker and job are irrevocably matched. For $A$ the set of matched jobs, and $\mu$ the edges these jobs are matched along, the platform's \emph{revenue} is $\sum_{j \in A} v_j - \sum_{(ij) \in \mu} \pi_{ij}$, where $v_j$ is the value of  completing  job $j \in J$.\footnote{It is also natural to study social welfare maximization, or maximization of convex combinations of the welfare and revenue. Our algorithms extend naturally to such objectives. See \Cref{sec:framework}.}

\medskip

 We assume that the platform learns the probability that workers accept a job at a certain price through multiple interactions with them. We further assume that due to the size of the market, this learning is robust to strategic behavior of the workers and we therefore ignore such issues. 

In our setting, the set of jobs and workers is fixed and known a priori. Even though that is rarely the case in practice, in many applications of interest the underlying network of between jobs and workers changes slowly. For example, in    applications like pet sitting or grocery shopping, the tasks are often scheduled in advance. In the case of food delivery, there is usually a 20-30 minute lag between the time the requests are made and the jobs are available, so it is practical to implement our sequential pricing mechanisms over a shorter planning horizon. 

In most settings, it is undesirable to overwhelm workers with too many queries. Therefore, we also consider the following generalization of our problem, borrowing the term \emph{patience} from the literature on stochastic probing problems (see \Cref{sec:related-work}). 

\paragraph{The Sequential Pricing Problem with Patience.} The setting is the same as above, but each worker $i$ has an associated \emph{patience} $\ell_i \in \mathbb{Z}_{\ge 0}$ and we can only offer prices along at most $\ell_i$ edges incident on $i$.

\medskip

It is not hard to see that the above sequential pricing problem is NP-hard even without patience constraints. In fact, using the reduction of Xiao et al.~\cite{xiao2020complexity}, one can show that this is the case already when the number of workers is one. 
 
With this challenge in mind, our goal is to design efficient pricing policies giving \emph{approximations} to the revenue achievable by the optimum sequential pricing algorithm. We briefly note that while in many online stochastic matching problems, the goal is to achieve a constant-factor approximation to the optimum \emph{offline} algorithm, such a result is impossible in our setting.\footnote{Consider a single edge $e$ with a job of value $k$, and an agent with a random cost $W \in [0, k-1]$ for completing the job, distributed as $\Pr[W \le w] = p_{ew}= 1/(k-w)$ for $0 \le w \le k-1$. By construction, any online algorithm gets revenue at most 1, while the optimum offline algorithm gets revenue $k - \mathbb{E}[W] = k - \int_{0}^{k-1} (1 - 1/(k-w)) \, dw =  \ln(k) + 1$.}

\paragraph{\textbf{Our Approach: Reducing to (and Improving) Contention Resolution Schemes}.}

Our approach for designing efficient approximate sequential pricing algorithms is to first solve a natural linear programming (LP) relaxation, with variables $\{ y_{ijw} \}$ representing the probability we offer price $w$ for worker $i$ to complete task $j$. We then round this LP in a (self-imposed) \emph{random-order online} fashion. 
This gives a common algorithmic framework for sequential pricing with and without patience constraints. 
We show that the analysis of instantiations of our algorithm can be seen as a reduction to random-order online algorithms for two well-studied problem in the literature: (i) online contention resolution schemes (OCRS), and (ii) stochastic OCRS (closely related to \emph{stochastic probing}).
More generally, we show that any chosen-order contention resolution schemes for matching result in improved bounds for our problems.

Our quantitative results follow precisely from improved algorithms for these well-studied contention resolution scheme questions, which we now introduce.

\subsection{Background: Contention Resolution Schemes}

When rounding a solution to some LP relaxation of some constrained problem, the most natural approach is to independently round its decision variables. For example, consider the following well-studied fractional relaxation of matching constraints:
\begin{align}\label{matching-deg-polytope}
\calP(G):=\left\{x\in \mathbb{R}^E_{\geq 0} \,\,\Bigg\vert\, \sum_{e\ni v}x_e\leq 1 \quad \forall v\in V\right\}.
\end{align}
Let $\vec{x} \in \calP(G)$ be a fractional solution to a matching problem with some linear objective $\vec{w}$. Rounding each edge $e$ independently will result in a random subset $R(x)\subseteq E$ of expected objective value equal to that of $\vec{w} \cdot \vec{x}$.
Unfortunately, this set $R(x)$ will most likely not form an integral matching. How much smaller (in expectation) can the highest-valued matching in $R(x)$ be compared to $\vec{w}\cdot\vec{x}$? This is precisely the \emph{correlation gap} of $\calP(G)$,  introduced in the context of stochastic optimization problems more generally in \cite{agrawal2012price}.

When rounding a point $\vec{x}\in \calP(G)$ (and more generally in a polytope containing the convex hull of some constraint set), especially in constrained computational settings, to obtain a feasible integral solution, one must in some sense resolve \emph{contention} between rounded edges (elements).
This motivates the definition of \emph{contention resolution schemes (CRS)}, introduced by \cite{chekuri2014submodular}, and later extended to online settings by \cite{feldman2016online,adamczyk2018random,lee2018optimal}. The following definition summarizes the notion of contention resolution schemes for matching.

\begin{Def}[Contention Resolution Scheme (CRS)]
    The initial input to a contention resolution scheme is a point $x\in \calP(G)$; each edge $e\in E$ is \emph{active} independently with probability $x_e$. The output of a CRS is a matching consisting of active edges.  A CRS is \emph{$c$-balanced} if its output $I$ satisfies
    $$\Pr[e\in I] \geq c\cdot x_e.$$ 
    A \emph{(Random-order) Online CRS}, or (RO-)OCRS for short, has elements revealed to it (in random-order) online, and must decide for each active element $e$ whether to add it to $I$ or discard it immediately and irrevocably on arrival.
\end{Def}

By linearity of expectation, a $c$-balanced CRS for relaxation $\calP$ of $\calF$ implies a lower bound of $c$ on this relaxation's correlation gap. The opposite direction is also true; Chekuri et al.~\cite{chekuri2014submodular} show using LP duality that a correlation gap of $c$ implies a $c$-balanced CRS. However, this proof is non-constructive, and does not imply a computationally efficient (i.e., polytime) CRS, let alone an efficient \emph{online} CRS.

Contention resolution schemes represent a core technique in many stochastic optimization problems, from (secretary) prophet inequalities, posted-price mechanisms and stochastic probing problems. (See \Cref{sec:related-work}.)
In this work, we observe the relevance of RO-OCRS, and stochastic CRS (see \Cref{sec:framework} for this latter notion) to our sequential pricing problem, and ask how much we can improve known bounds for such contention resolution schemes.

\subsection{Our Contributions}

Given our reductions from the sequential pricing problem (with patience) to (stochastic) RO-OCRS, we first state our results in the context of these latter two well-studied problems.

\paragraph{RO-OCRS.}

The classic work of Karp and Sipser \cite{karp1981maximum} implies an upper bound of $0.544$ on the correlation gap---and hence offline CRS---for (bipartite) matchings (see \cite{guruganesh2018understanding,bruggmann2020optimal}).
Numerous prior works imply or present explicit lower bounds on this gap (see \Cref{sec:related-work}).
The highest known lower bound on the correlation gap is $0.4326$ \cite{bruggmann2020optimal}, and the best RO-OCRS has a similar balance ratio of $\frac{1}{2}(1-\e^{-2})\approx 0.4323$~\cite{brubach2021improved}.
For bipartite graphs, the best offline CRS is $0.4762$-balanced \cite{bruggmann2020optimal}, but the best RO-OCRS is only $\frac{1}{2}(1-\e^{-2})\approx0.4323$-balanced \cite{brubach2021improved}.

\medskip

Our first result is a significant improvement on both these latter bounds for RO-OCRS.

\begin{wrapper}
\begin{restatable}{thm}{THMROOCSR}\label{thm:ro-ocrs}
    There exists a (polytime) RO-OCRS which is $\corrgap$-balanced for the matching problem and $\bipcorrgap$-balanced for the matching problem in bipartite graphs.
\end{restatable}
\end{wrapper}

Our improvement of $\lift$ over the previous best correlation gap lower bound of $0.4326$ \cite{bruggmann2020optimal} can be contrasted with the latter work's improvement of $0.0003$ over the $\frac{1}{2}(1-\e^{-2})\approx 0.4323$ bound of \cite{guruganesh2018understanding}. 
That the correlation gap can be improved this much using \emph{random-order online} CRS is perhaps surprising, since, as pointed out by Bruggmann and Zenklusen \cite{bruggmann2020optimal}, ``the online random order model has significantly less information than the classical offline setting [they] consider, [and therefore] it is harder to obtain strong balancednesses in the online setting.''

\paragraph{Stochastic probing with patience.}
In many applications, verifying whether an edge is active incurs some cost (of energy or concentration) for one or both of its endpoints.
This is captured by stochastic matching problems with \emph{patience constraints}, studied by \cite{chen2009approximating,adamczyk2011improved,bansal2012lp,adamczyk2015improved,baveja2018improved,adamczyk2020improved,brubach2021improved}. For this problem, we improve on the state-of-the-art $0.382$ bound \cite{brubach2021improved} via a new \emph{stochastic OCRS} (see \Cref{sec:framework} for a definition of this natural generalization of OCRS).

\begin{wrapper}
\begin{restatable}{thm}{patiencetheorem}\label{thm: patience-theorem}
    There exists a $\patiencratio$-balanced patience-constrained stochastic RO-OCRS for the matching problem. For bipartite matching with patience constraints on only a single side, this stochastic RO-OCRS is $\bippatienceratio$-balanced. 
\end{restatable}
\end{wrapper}

\Cref{thm: patience-theorem} directly implies $0.395$-approximate and $0.426$-approximate algorithms in the corresponding stochastic probing settings. Also, as a direct consequence of both the OCRSes presented in \Cref{thm:ro-ocrs} and \Cref{thm: patience-theorem}, we have the following result for our original pricing problem. 

\begin{wrapper}
\begin{cor}
    There exists a $\bipcorrgap$-approximate efficient algorithm for the Sequential Pricing Problem, and a $\bippatienceratio$-approximate algorithm for the Sequential Pricing Problem with patience.
\end{cor}
\end{wrapper}

Finally, we introduce a random-order vertex arrival model with \emph{uncorrelated} edges  (breaking with recent work assuming at most one edge is active per arriving vertex \cite{ezra2020online,fu2021random}), and provide a $(1-1/e)^2\approx 0.3996$-balanced OCRS for random-order vertex arrivals in \Cref{sec:ro-vertex}.

\subsection{Techniques}

An approach used in the literature for both random-order OCRS and stochastic probing is a natural one: 
when considering an edge $e$, in order to allow both endpoints a fair chance of being matched to later arriving edges, one randomly 
decreases, or \emph{attenuates}, the probability with which $e$ is matched on arrival. This approach was highly effective for these problems \cite{baveja2018improved,ehsani2018prophet,lee2018optimal,brubach2020attenuate,brubach2021improved}, and underlies the state-of-the-art for both \cite{brubach2021improved}.
To lower bound the probability that an edge $e$ is matchable (i.e., its endpoints are free) upon arrival, these prior works lower bound the probability that all edges in $e$'s neighborhood, denoted by $N_e$, are either inactive, arrive after $e$, or are rejected by the attenuation step. 
Note that if edges incident on $e$ have no other edges in their neighborhood (e.g., if $e$ is the central edge in a path of length three), then this analysis is tight. Indeed, similar examples rule out analysis improving on the natural (and previous best) bound of $\frac{1}{2}(1 - \e^{-2})$ using all prior attenuation functions. 

Our key observation is that in scenarios as the above, while $e$ is not matched with probability better than $\frac{1}{2}(1 - \e^{-2})$ times $x_e$, its incident edges are, and we can therefore attenuate these more, while still beating this bound. If, on the other hand, edges incident on $e$ do have many incident edges, and more ``contention'' for their endpoints, then there are other events in which $e$ can be matched: for example, some single incident $f\in N_e$ is active, precedes $e$, and passes the attenuation step, but is nonetheless not matched, since some previous edge $g\in N_f$ is matched.
An appropriately-chosen attenuation function which penalizes edges with less contention therefore allows us to capitalize on whichever of these scenarios holds for $e$, and ultimately break the previous state-of-the-art balance ratio of $\frac{1}{2}(1 - \e^{-2})$.
A similar approach underlies our result for the patience-constrained problem.

\subsection{Further related work}\label{sec:related-work}

\paragraph{(RO-O)CRS for matching.} Illustrating the ubiquity of problems addressed by CRSes and matching problems,
the literature is rich in results which either explicitly or implicitly provide non-trivial lower bounds on the correlation gap---and hence balance ratio of CRS---for matchings (e.g., \cite{yan2011mechanism,cygan2013sell,chekuri2014submodular,feldman2016online,guruganesh2018understanding,gravin2019online,bruggmann2020optimal,ezra2020online,brubach2021improved}). 
The best previous CRS for matching is due to Bruggmann and Zenklusen,  who provide a monotone CRS, which allows them to give result for  \emph{submodular} objectives, though this result requires them to use the full (exponential-sized) characterization of the matching polytope \cite{edmonds1965maximum}.
OCRSes for matching have also been designed in the \emph{vertex-arrival} setting \cite{ezra2020online, fu2021random}, although as we note in \Cref{sec:framework} these don't directly imply results for our edge-by-edge pricing problem or for the correlation gap, as they assume every arriving vertex has at most one incident active edge.

\paragraph{Online stochastic probing.} Our sequential pricing problem is closely related to the stochastic probing literature. In the most general formulation, a set of elements $E$ is given where element $e$ has weight $w_e$ and is active with probability $p_e$. Elements reveal their active status after being probed, and we can \emph{probe} elements according to an outer feasibility constraint $\mathcal{I}_{out}$ and can accept active elements according to an inner feasibility constraint $\mathcal{I}_{in}$. Probed active elements must be accepted. Stochastic probing has been studied in an extremely general settings where $\mathcal{I}_{in}$ and $\mathcal{I}_{out}$ are intersections of matroids and knapsacks, and further generalizations \cite{gupta2013stochastic, gupta2016algorithms, gupta2017adaptivity}. Further research has extended the linear objective to submodular functions \cite{adamczyk2016submodular}. 

A very well-studied special case concerns the setting where $E$ is the set of edges in a graph, $\mathcal{I}_{in}$ is the set of matchings, and $\mathcal{I}_{out}$ specifies that each vertex $v$ has some \emph{patience} $\ell_v$, corresponding to the maximum number of edges of this vertex that the algorithm can probe.
This is motivated for example by food delivery services, where couriers should not be contacted more than a certain number of times without being commissioned to pick up any order.
This problem was studied in an offline setting in \cite{chen2009approximating,adamczyk2011improved,bansal2012lp,adamczyk2020improved}, and in an online setting \cite{bansal2012lp,adamczyk2015improved,baveja2018improved,brubach2021improved}. Like in \cite{brubach2021improved}, we give improvements for \emph{offline} stochastic matching even if edges arrive online in a uniformly random order. Without patience constraints, \cite{gamlath2019beating} beat the $0.5$-approximate greedy algorithm and gave a $(1-1/\e)$-approximation. The difficulty in applying the method used in \cite{gamlath2019beating} to our setting is that we have multiple prices for each edge; in their setting, distributions are Bernoulli and approximations to the offline optimum can be obtained. However, in our setting, no such approximation exists (as noted earlier). The same can be said for the weighted query-commit problem for matching studied in \cite{fu2021random}. 

\smallskip 

\paragraph{Posted-Price Mechanisms.} Our work is also closely related to the study of incentive-compatible posted-price mechanisms as in \cite{chawla2010multi, kleinberg2019matroid}, which develop approximations to the revenue of the optimal deterministic truthful mechanism via a reduction to prophet inequalities. (See \cite{lucier2017economic} for a survey of this large area of research, much of which is inspired by this reduction.) The prophet inequalities of \cite{gravin2019prophet} and \cite{ezra2020online} hence imply $\nicefrac{1}{3}$- and $0.337$-approximate algorithms respectively for revenue in this setting. 

There are subtle but important differences between our pricing problem and this line of work stemming from the choice of benchmark. \cite{chawla2010multi} and \cite{kleinberg2019matroid} explicitly concern themselves with incentive-compatibility and hence take the optimal truthful mechanism as their benchmark. It is worth noting that to implement their incentive-compatible mechanisms (and others based on reductions from prophet inequalities), it is important to communicate to agents information about the state of the \emph{entire} market. In extremely large gig economy platforms, this is not feasible; we additionally assume that based on a wealth of data large online marketplaces have learned probabilities $\{p_{ijw}\}$ in way that is robust to strategic behavior. Hence, we compete with the stricter benchmark of the optimal (not necessarily truthful) posted-price policy.

\section{The Algorithmic Template Via Contention Resolution}\label{sec:framework}

In this section, we describe formally our approach for the sequential pricing problem (with and without patience constraints), via a connection with OCRS.

Many natural greedy algorithms for this problem fail (see \Cref{app:rulingout}). Instead, our approach is to write a natural linear program relaxation (LP-Pricing) and round it in a (self-imposed) random-order online fashion. Recall that for edge $e = (ij)$, $p_{ijw} = p_{ew}$ is notation for the probability that worker $i$ accepts job $j$ at price $w$. 

\begin{align}
    \tag{\textbf{LP-Pricing}}  \label{LP-main} \quad \max & \sum_{w, e=(ij)} y_{ew}\cdot p_{ew}\cdot (v_j - w) \\
     \textrm{s.t.} & \sum_{w} y_{ew}  \leq 1 \qquad \qquad\qquad\, \forall e \label{eqn:single-probe} \\
    & \sum_{e\ni v}\sum_{w} y_{ew}\cdot p_{ew}  \leq 1  \qquad \,\, \forall v \label{eqn:single-success-LP} 
    \\
    & \sum_{e\ni v}\sum_{w} y_{ew}  \leq \ell_v  \qquad\qquad \,\, \forall v \label{eqn:patience-LP2}\\
    & y_{ew}  \geq 0 \qquad \qquad \qquad\qquad \forall e,w. \label{eqn:positivity}
\end{align}

For any sequential pricing algorithm $\calA$, the vector $\vec{y}$ with $y_{ew}$ equal to the probability that $\calA$ offers price $w$ along edge $e$ is a feasible solution to \ref{LP-main}. Indeed, $\calA$ can query at most one weight for each edge in expectation \eqref{eqn:single-probe}, successfully matches at most one edge incident on any vertex $v$ in expectation \eqref{eqn:single-success-LP}, and queries at most $\ell_v$ edges incident on $v$ in expectation \eqref{eqn:patience-LP2}. Moreover, since $y_{ew}$ are probabilities, we immediately have \eqref{eqn:positivity}. Therefore, since the LP objective of $\vec{y}$ precisely corresponds to the expected revenue of $\calA$, we find that \ref{LP-main} upper bounds the expected revenue of any sequential pricing algorithm $\calA$, as summarized in the following lemma.

\begin{lem}\label{LP>=OPT} The optimal value of \ref{LP-main} upper bounds the expected revenue attainable by the optimal algorithm for the sequential pricing problem with patience. 
\end{lem}

We note that a similar LP relaxation can be written if our objective is the maximization of social welfare. In this case we assume that for each worker $i$ and job $j$ we are given a random \emph{cost} to the worker $c_{ij}$ for completing $j$, and agreeing to perform job $j$ if the price exceeds its cost, i.e., $\Pr[c_{ij} \le w] = p_{ijw}$ for every $w$. Then, to optimize social welfare, we can replace our objective with $$ \max \sum_{w, e=(ij)} y_{ew}\cdot  p_{ew} \cdot \left(v_j - \mathbb{E}[c_{ij} \mid c_{ij} \le w] \right).$$ We could similarly optimize over a convex combination of revenue and social welfare.

To round \ref{LP-main} we consider edges in  random order; in particular, we have each edge $e$ sample a random ``arrival'' time $t_e \sim \text{Unif}[0,1]$, with edges arriving in increasing order of $t_e$. 
For each edge $e$, we further set a price $\pi_e$ equal to each $w$ with probability $y_{ew}$. (This step is well-defined, by \eqref{eqn:single-probe}.)
Now, if the edge is free (both of its endpoints are unmatched) and both of its endpoints have remaining patience, we propose this price $\pi_e$ along $e$ with probability $a(e)$.
The algorithm's pseudocode is given in \Cref{alg:main}.

\begin{algorithm}
\SetAlgoLined

         Solve \ref{LP-main} 
         
        \textbf{for} each edge $e$, set price $\pi_e$ at random so $\Pr[\pi_e = w] = y_{ew}$
        
        \For{each edge $e$, in increasing order of $t_e\sim \text{Unif}[0,1]$}{
           
            \If{$e$ is free and both endpoints have remaining patience}{
    		        {\textbf{with probability} $a(e)$, propose price $\pi_e$ along $e$ (and match if accepted) \label{line:proposal}} 
	        }
	    }
    \caption{Sequential Pricing Algorithm}
    \label{alg:main}
\end{algorithm}

Instantiations of \Cref{alg:main} differ by their choice of \emph{attenuation function} $a(e)$. (Following the terminology of  \cite{brubach2020attenuate,brubach2021improved,baveja2018improved}.) 
Before describing one simple such function whose performance will serve as a baseline, we first note that for the sequential pricing problem (without patience constraints), \Cref{alg:main} can be seen as a reduction from this problem to the RO-OCRS described in \Cref{alg:main-ocrs}, which generalizes some prior RO-OCRSes for matching \cite{lee2018optimal,brubach2021improved,bruggmann2020optimal}.

\begin{algorithm}
\SetAlgoLined
  
         \textbf{for} each edge $e$, sample a random arrival time $t_e \sim \text{Unif}[0,1]$

        \For{each edge $e$, in increasing order of $t_e \sim \text{Unif}[0,1]$}{
   
            \If{$e$ is active and is free}{
    		        \textbf{with probability} $a(e)$, match $e$\label{line:RO-OCRS-test}
	        }
	    }
	
    \caption{RO-OCRS}
    \label{alg:main-ocrs}
\end{algorithm}

\begin{lem}\label{lemma:reduction}
    If \Cref{alg:main-ocrs} with attenuation function $a(\cdot)$ is $c$-balanced, then \Cref{alg:main} with function $a(\cdot)$ is $c$-approximate for the sequential pricing problem (without patience constraints).
\end{lem}
\begin{proof}
Suppose we sample for each edge $e=(ij)$ and price $w$ independent Bernoulli variables $A_{ew} \sim \text{Ber}(p_{ew})$  \emph{in advance}.
Since the random choice of the weight $\pi_e$ and the probability that worker accepts $\pi_e$ are both independent of the other random choices of \Cref{alg:main}, the output of this new algorithm is the same as that of \Cref{alg:main-ocrs}, where the (random) set of active edges is taken to be $\{e=(ij) \mid A_{e\pi_e} = 1\}$.
Since each edge $e$ is active independently with probability $x_e := \sum_{w} y_{ew}\cdot p_{ew}$, by constraints \eqref{eqn:single-success-LP} and \eqref{eqn:positivity} we have that $\vec{x}\in \calP(G)$, and so this fractional matching (and the corresponding random set of active edges) is indeed a valid input to the RO-OCRS \Cref{alg:main-ocrs}. It remains to relate the balancedness of this RO-OCRS to the revenue of \Cref{alg:main}.
 
Fix an edge $e=(ij)$. By the $c$-balancedness of \Cref{alg:main-ocrs}, we know that $\Pr[e\in \mu] = c'\cdot x_e$, for some $c'\geq c$.
Recall that $e$ is matched if and only if it is active ($A_{e\pi_e} = 1$) and it is both free at time $t_e$ and an independent Bernoulli variable $\text{Ber}(a(e))$ comes up heads. Since these latter two events are independent of $e$'s active status, we have that $\Pr[e \in \mu \mid e \text{ active}] = c'$.
Indeed, since the joint event that $e$ is free before time $t_e$ and $\text{Ber}(a(e))=1$  are both independent of $\pi_e$ and all the $A_{ew}$ variables, we have that for each $w$,  
$$\Pr[e \in \mu \mid A_{e\pi_e} = 1] = c'.$$

Consequently, the expected revenue satisfies the following:
\begin{align*}
    \E \left[ \sum_{j\in A} v_j - \sum_{e\in \mu} \pi_e \right] & = \sum_{ij}\sum_w (v_j - w)\cdot c'\cdot \Pr[\pi_e = w] \cdot \Pr[A_{e\pi_e} = 1 \mid \pi_e = w] \\
    & = \sum_{ij}\sum_w (v_j - w)\cdot c'\cdot y_{ew}\cdot p_{ew} \\
    & = c'\cdot OPT(\text{\textbf{LP-Pricing}}) \geq c\cdot OPT(\text{\textbf{LP-Pricing}}).
\end{align*}
Combining the above with \Cref{LP>=OPT}, the lemma follows.
\end{proof}

The argument of \Cref{lemma:reduction} can be used to show a general black-box reduction from the sequential pricing problem to \emph{chosen-order} OCRS, where we have full power to decide the order in which edges arrive. 
\begin{lem}
A $c$-balanced chosen-order OCRS for matchings implies a  $c$-approximate algorithm for the sequential pricing problem (without patience constraints).
\end{lem}

We note that a similar reduction does \emph{not} work for vertex-arrival OCRS, where \cite{ezra2020online} show that a $\nicefrac{1}{2}$-balanced scheme exists, and \cite{fu2021random} show that an $\nicefrac{8}{15} \approx 0.533$-balanced scheme exists when vertices arrive in random order. The issue in trying to appeal to these results is that they operate in the \emph{batched} setting where the online algorithm knows the realization of edges incident on an arriving vertex in advance. Furthermore, it is assumed at most one of the edges incident on each arriving vertex is active. These constraints mean that OCRSes for the vertex-arrival setting do not directly have a connection to our problem, where edges must be probed one-by-one.

\Cref{lemma:reduction} allows us to analyze  \Cref{alg:main} using the simpler \Cref{alg:main-ocrs}, avoiding some notational clutter, and so this is the terminology we will use when analyzing our algorithm without patience constraints, beginning with the next section.

When studying \Cref{alg:main} with patience constraints, a very similar reduction holds, but we must work with the \emph{stochastic} OCRS setting as in \cite[Definition 2]{adamczyk2015non}. We give the special case of this general definition for our setting of matching with patience constraints. We say a vertex $v$ \emph{has remaining patience} or just \emph{has patience} if it has been queried less than $\ell_v$ times, and \emph{has lost patience} otherwise. 

\begin{Def}[Stochastic OCRS for Matching with Patience]\label{def:stochastic-ocrs}
    We are given a graph $G = (V,E)$  and patiences $\{\ell_v\}_{v \in V}$, along with vectors $\vec{y}, \vec{p} \in [0,1]^E$. For every vertex $v$ we know $\sum_{e \ni v} y_e \le \ell_v$ and $\sum_{e \ni v} y_e \cdot p_e \le 1$. Each edge $e \in E$ is \emph{active} independently with probability $p_e$.
    
    A stochastic OCRS processes edges in some order; when processing an edge $e$, if $e$'s endpoints are unmatched and have remaining patience, it can decide to \emph{probe} $e$. Any edge that is probed is active independently with probability $p_e$, and if active must be matched. The stochastic OCRS is \emph{$c$-balanced} if for every edge $e$
    $$\Pr[e \in \mu] \geq c\cdot y_e \cdot p_e,$$
    where $\mu$ is the outputted matching. In a \emph{Stochastic RO-OCRS}, the edges are processed in a uniformly random order. 
\end{Def}

\begin{algorithm}[h]
\SetAlgoLined
   
        \textbf{for} each edge $e$, sample a random arrival time $t_e \sim \text{Unif}[0,1]$

        \For{each edge $e$, in increasing order of $t_e \sim \text{Unif}[0,1]$}{
    
            \If{$e$ is is free, and its endpoints have remaining patience}{
    		        \textbf{with probability} $y_e \cdot a(e)$, probe $e$ (and match, if active) \label{line:RO-OCRS-test-stochastic}
	        }
	    }
	
    \caption{Stochastic RO-OCRS}
    \label{alg:stochastic-ocrs}
\end{algorithm}

\Cref{alg:stochastic-ocrs} captures our algorithmic template for stochastic RO-OCRS. 

In the patience setting, we have the following analogous claim to \Cref{lemma:reduction}. 

\begin{restatable}{lem}{reductionpatience} \label{lemma:reduction-patience}
    If \Cref{alg:stochastic-ocrs} with attenuation function $a(\cdot)$ is $c$-balanced, then \Cref{alg:main} with function $a(\cdot)$ is $c$-approximate for the sequential pricing problem with patience constraints.
\end{restatable}

The proof proceeds similarly to the proof of \Cref{lemma:reduction}; in particular, given a solution $\{y_{ew}, p_{ew}\}$ to \ref{LP-main}, we note that taking $y_e := \sum_w y_{ew}$ and $p_e := \frac{\sum_w y_{ew} p_{ew}}{\sum_w y_{ew}}$ gives a valid input to a stochastic OCRS algorithm. The full proof is deferred to \Cref{app:omittedproofs2}. As one might expect, the proof also applies to reducing to the more general setting of \emph{chosen-order} stochastic OCRS.

\begin{lem}
A $c$-balanced chosen-order stochastic OCRS for matchings with patience implies a  $c$-approximate algorithm for the sequential pricing problem with patience constraints. 
\end{lem}

\subsection{Warm-Up: A $\frac{1}{2}(1-\e^{-2})$-balanced RO-OCRS}
\label{sec:warm-up-RO-edge-arrival} 

We start by presenting a simple instantiation of \Cref{alg:main-ocrs} yielding a $\frac{1}{2}(1-\e^{-2})\approx 0.432$-balanced RO-OCRS for matching in general graphs.
This first attenuation funtion we use is implied by the work of Lee and Singla \cite[Appendix A.2 of arXiv version]{lee2018optimal}. They considered the special case of star graphs, for which their function gives an optimal $(1-1/\e)$-balanced RO-OCRS. We show that the same attenuation function results in a $\frac{1}{2}(1-\e^{-2})$-balanced RO-OCRS in general graphs.

For our analysis of instantiations of \Cref{alg:main-ocrs}, we will imagine that the Bernoulli($a(e)$) variable used in \Cref{line:RO-OCRS-test} was drawn before testing if $e$ is free, and say that an edge $e$ which is both active and has its Bernoulli variable equal to one (i.e., it will pass the test in \Cref{line:RO-OCRS-test} if it is active and free) is \emph{realized}.
Note that conditioned on the arrival time $t_e=t$, the realization of edge $e$ and its matched status upon arrival time (i.e., whether or not it is free) are independent. 

The following lemma lower bounds the probability of an edge being in the output matching $\mu$ by considering the event that no edge $f \in N_e$ (edges incident on $e$) arrives before $t_e$ and is realized.

\begin{lem}\label{lem:analysisOfAlg2}
    \Cref{alg:main} run with $a(e) = a_1(x_e,t_e) := \e^{-t_e x_e}$ guarantees for each edge $e\in E$ that
    $$\Pr[e \in \mu] \geq \frac{1}{2} \left( 1- \frac{1}{\e^{2}} \right) \cdot x_{e} \ge 0.432 \cdot x_{e}.$$
\end{lem}
\begin{proof}
    Let $R_0(e)$ be the event that zero edges $f\in N_e$ are realized prior to time $t_e$. Clearly, $e$ is free upon arrival if $R_0(e)$ occurs.
   
    Using independence of arrival times and (in)activity of edges $f\in N_e$, and the fact that $\vec{x}\in \calP(G)$, we can lower bound the probability of $R_0(e)$.
    \begin{equation}
        \Pr[R_0(e) \mid t_e = y] = \prod_{f\in N_e} \left( 1 - \int_0^{y} x_f \cdot \e^{- x_f z} \, dz \right)  =  \prod_{f\in N_e} \e^{-x_f y} = \e^{ -{ \sum_{f \in N_e} x_f y} } \geq \e^{- (2 - 2x_e) y}.  \label{ineq:noblockingedge}
    \end{equation}
    
    \Cref{ineq:noblockingedge} implies that $\Pr[e \textrm{ free} \mid t_e = y]\geq \Pr[R_0(e) \mid t_e=y]\geq \e^{-(2-2x_e)y}$.
    Taking total probability over arrival time $t_e$, and noting that
    $e$ is matched if $e$ is both realized and free upon arrival (with these probabilities independent when conditioned on $t_e=a$), the lemma follows.
    \begin{align}
        \nonumber \Pr[e \in \mu]  &\ge \left( \int_{0}^1 \e^{ - (2 - 2x_e) y }  \cdot \e^{- x_e y}  \, dy \right) \cdot x_e \ge  \left( \int_{0}^1 \e^{ - 2y}  \, dy \right) \cdot x_e  = \frac{1}{2} \left( 1- \frac{1}{\e^{2}} \right) \cdot x_e. \qedhere
    \end{align}
\end{proof}

\paragraph{Tightness of Analysis.}\label{sec:tightness-warm-up}
The above analysis is tight for \Cref{alg:main} with the above attenuation function:
consider a 3-edge path with edges $(a, b, c)$ with $x_a = x_c = 1 - 1/n$ and $x_b = 1/n$, and consider the bound on the probability that edge $b$ is matched. On this instance, all inequalities in the above analysis are tight as $n \rightarrow \infty$, since $b$ is matched if and only if $R_0(b)$ occurs (if $a$ or $c$ are realized before $b$, they must be matched), the fractional degree of both endpoints of $b$ is 1, and $x_b \rightarrow 0$. So, the probability $b$ is matched approaches $(1 - \e^{-2})/2 \cdot x_b$ as $n \rightarrow \infty$.\footnote{Note that consequently, the overall LP rounding scheme via this RO-OCRS does not provide better than a $(1 - \e^{-2})/2$-approximation if $b$ has much larger weight than $a$ or $c$.}

\begin{rem}
    The same attenuation function $a_1(\cdot )$ allows to recreate the $0.382$-approximation for online stochastic probing due to Brubach et al.~\cite{brubach2021improved}. See \Cref{sec: patience}. In what follows we present an improved attenuation function which will improve the bounds for both problems we consider.
\end{rem}

\subsection{Motivation for the Improved Attenuation Function}

To analyze instantiations of algorithms \ref{alg:main} and \ref{alg:main-ocrs}, we need to analyze the probability that $e$ is free (at time $t_e$). 
In the proof of \Cref{lem:analysisOfAlg2}, we lower bounded this event by the event that zero edges incident on $e$ were realized before time $t_e$, which we denoted by $R_0(e)$. Generalizing this notation further, if we let $Q(e):=|\{f\in N_e \mid f \textrm{ realized} \wedge (t_f < t_e)\}|$ be the set of realized edges incident on $e$ arriving before $e$, we define the following events for each non-negative integer $k$.
$$R_k(e) := \mathds{1}[ |Q(e)| = k ].$$
Since the events $\{R_k(e) \mid k\in \mathbb{Z}_{\geq 0}\}$ partition our probability space, we have
\[
\Pr[e \textrm{ free}] = \sum_{k=0}^\infty \Pr[e \textrm{ free} \wedge R_k(e)].
\]

In \Cref{sec:warm-up-RO-edge-arrival} we lower bounded $\Pr[e \textrm{ free}]$ by lower bounding $\Pr[R_0(e)]=\Pr[e \textrm{ free} \wedge R_0(e)]$. A natural approach to improve the above is to likewise lower bound $\Pr[e \textrm{ free} \wedge R_k(e)]$ for $k>0$. Unfortunately, 
as illustrated by the tight example in \Cref{sec:warm-up-RO-edge-arrival}, 
such terms can be zero if all edges $f\in N_e$ have no edges incident on them other than $e$. Indeed, if such an edge $f$ arrives at time $t_f<t_e$ and is realized, it will be matched, and $e$ will consequently not be free upon arrival. Consequently, $\Pr[e \textrm{ free} \wedge R_k(e)]=0$ for all $k>0$ if $N_f=\{e\}$ for all $f\in N_e$.

Our key observation is that edges $f\in N_e$ with few incident edges other than $e$ (more precisely, with low $\sum_{g\in N_f} x_g$) are matched with probability strictly greater than $\frac{1}{2}(1-\e^{-2})\cdot x_f$. 
For example, if $\sum_{g\in N_f} x_g \leq 1$, then $f$ is matched with probability $(1-1/\e)\cdot x_f \approx 0.632 \cdot x_f$. (This is precisely the analysis of \cite{lee2018optimal} for star graphs.)
We can therefore safely attenuate such edges $f$ more aggressively and still beat the bound of $\frac{1}{2}(1-\e^{-2})\approx 0.432$.

The upshot of this approach is that, informally, if $N_e$ consists mostly of edges $f$ with little ``contention'', i.e., low $\sum_{g\in N_f} x_g$, then this increases the probability that no edges incident on $e$ are realized before time $t_e$, i.e., it increases $\Pr[e \textrm{ free} \wedge R_0(e)]$. In contrast, if most incident edges $f\in N_e$ are highly contended, i.e. have a high value of $\sum_{g\in N_f} x_g$, we will show that this implies a constant lower bound for $\Pr[e \textrm{ free} \wedge R_1(e)]$. In particular, we show that if a single edge $f\in N_e$ is realized and arrives at time $t_f<t_e$,
then there is a non-negligible probability that some edge $g\in N_f$ is matched before time $t_f$, and so $e$ is free despite having a realized incident edge $f$. Combining the above yields our improved bounds for $\Pr[e \textrm{ free}]$, and consequently for $\Pr[e \in \mu]$.

\paragraph{The tighter attenuation function.} Motivated by the above discussion, we consider an attenuation function that is more aggressive (i.e., is lower) for edges $e$ with small $d_e := \sum_{f\in N_e} x_f$. More precisely, we attenuate more aggressively edges $e$ with large \emph{slack}, $s_e:=2-d_e-x_e.$ (Note that $s_e \geq x_e\geq 0$, since $\vec{x}\in \calP(G)$.) In particular, we will use the following attenuation function.
\[
a_2(t_e,x_e,s_e) := \e^{-t_e x_e} \cdot (1 - \alpha \cdot s_e)
\]
Here, $\alpha \in [0, 0.5]$ is a constant we will optimize over at the end of our analysis. (Note that since $s_e \in [0,2]$, the term $a_2(e) \in [0,1]$ is a valid probability.)

\paragraph{Optimizing Our Approach.} As we shall see in the next sections, the above choice of attenuation function allows to beat the state-of-the-art for RO-OCRS (and for the correlation gap), as well as stochastic RO-OCRS. However, it is by no means the only one we could have chosen to achieve this goal. That being said, our analysis requires several non-trivial cancellations and lower bounds given by our explicit choice of $a_2(\cdot,\cdot,\cdot)$, and so optimizing this function symbolically seems challenging. We leave the optimization of our approach as the topic of future research. 

\section{RO-OCRS for Matching: Going Beyond $\frac{1}{2}(1 - \e^{-2})$}\label{sec:going-beyond}

In this section we analyze \Cref{alg:main-ocrs} using the attenuation function $a_2(e)$, with which we substantially improve on the bound of $\frac{1}{2}(1-\e^{-2})$ for RO-OCRS. We begin with an overview of our approach. 

\subsection{Overview}\label{sec:roadmap}

The following function will appear prominently in our analysis:
$$h(x) := \frac{1}{x} \left( 1 - \e^{-x} \right) = \int_{0}^1 \e^{-x\cdot z} \,dz.$$ 

Note that we wish to prove that $\Pr[e\in \mu]/x_e$ is strictly greater than $h(2)=\frac{1}{2}(1-\e^{-2})$.
The following lemma will yield an improvement over this bound if most edges in $N_e$ have little contention. 

\begin{restatable}{lem}{saturatedbound} \label{lem:saturated}
For each edge $e\in E$, we have 
$$\Pr[ e \in \mu \wedge R_0(e)] \ge(1 - \alpha \cdot s_e)  \cdot \left(h(2) + 0.14\cdot \left( s_e + \alpha \sum_{f \in N_e} x_f s_f \right)\right) \cdot x_e.$$
\end{restatable}

We prove \Cref{lem:saturated} in \Cref{sec:a-0-lower}. For now, we provide some motivation for this lemma. Recall that $h(2)$ is precisely the bound we are trying to beat. Since $h(x)$ is the integral of a decreasing function in $x$, it is itself decreasing in $x$. 
We will set $\alpha$ to guarantee that if $s_e$ or $\sum_{f \in N_e} x_f s_f$ are large, then we already have a greater than $h(2)\cdot x_e \approx 0.432\cdot x_e$ probability that $e$ is matched. 

The technical meat of our analysis will be in providing improved bounds on $\Pr[e\in \mu \wedge R_1(e)]$ in the complementary scenario, where $s_e$ and $\sum_{f\in N_e} x_f s_f$ are both small. 
We address this in \Cref{sec: a1-lower-bound}. For reasons that will become apparent in that section, we need the following notation.
$$m_e := \max \{ x_f \mid f \in N_e, \text{ there exists a triangle containing } e, f \}.$$
Note that for bipartite graphs, $m_e=0$, though in general graphs this need not be the case.
Our main technical lemma is the following.

\begin{restatable}{lem}{unsaturatedbound}\label{lem:unsaturated}
    For each edge $e\in E$, we have
     $$\Pr[e \in \mu \wedge R_1(e) ] \ge \left( (1 - \alpha \cdot s_e)  (1-  2\alpha)^2  \cdot \sum_{f \in N_e}  x_f (1 - m_e - x_f - s_f)^+ \cdot 0.0275  \right) \cdot x_e.$$
\end{restatable}

In \Cref{sec:five-var} we combine both lemmas \ref{lem:saturated} and \ref{lem:unsaturated} to obtain a five-variable program, using which we can prove the following bound on our algorithm's balancedness using $\alpha=0.171$.

\THMROOCSR*

\subsection{Lower Bounding $\Pr[e \text{ \normalfont matched} \wedge R_0(e)]$}\label{sec:a-0-lower}

Before proving \Cref{lem:saturated}, we state the following facts useful to our bounding. 
\begin{fact} \label{fact:concavity}
For any $a, x \in [0,1]$, we have $x (1 - \e^{-a}) \le 1 - \e^{-ax}$.
\end{fact}
\begin{proof}
Note the function $k_a(x):=\frac{1-\e^{-ax}}{1-\e^{-a}}$ for $a\in [0,1]$ is concave; as $k_a(0)=0$ and $k_a(1)=1$ we have that $k_a(x)\geq x$ for all $x\in [0,1]$.
\end{proof}

\begin{proof}[Proof of \Cref{lem:saturated}]

As in the proof of \Cref{lem:analysisOfAlg2}, we condition on a specific arrival time $t_e$ and explicitly compute the probability of the event that no realized edges arrived before $e$ (the only change is the updated attenuation function). 
\begin{align*}
         \Pr[R_0(e) \mid t_e = y] &= \prod_{f\in N_e} \left( 1 - \int_0^{y} x_f \cdot  \e^{- x_f z} \cdot (1 - \alpha \cdot s_f) \, dz \right) \\
         &= \prod_{f\in N_e} \left( 1 - (1 - \alpha \cdot s_f)\cdot(1 - \e^{-x_f y})  \right) \\
        &\ge \prod_{f\in N_e} \left( 1 - (1 - \e^{-x_f y (1 - \alpha \cdot s_f)} ) \right) \\
        & = \e^{ -y( d_e - \alpha \cdot \sum_{f \in N_e} x_f s_f)},
\end{align*}
where the above inequality follows from \Cref{fact:concavity} applied to $a=(1-\alpha\cdot s_f)\in [0,1]$ and $x=x_f\in [0,1]$.

Integrating over the arrival time of $e$ yields the claimed bound: 
    \begin{align}
        \nonumber \Pr[e \in \mu \wedge R_0(e) ]  &= \Pr[R_0(e) \wedge e \text{ realized}]
        \\
        \nonumber &\ge \left( \int_{0}^1 \e^{ -y( d_e - \alpha\cdot \sum_{f \in N_e} x_f s_f)}  \cdot \e^{- x_e y} \cdot (1 - \alpha \cdot s_e)  \, dy \right) \cdot x_e & \textrm{def. } a_2(\cdot,\cdot,\cdot) \\
        \nonumber &=  (1 - \alpha \cdot s_e)  \cdot h \left(d_e + x_e - \alpha \sum_{f \in N_e} x_f s_f \right) \cdot x_e & \text{def. $h(\cdot)$}\\ 
        \nonumber &=  (1 - \alpha \cdot s_e)  \cdot h \left(2 - s_e - \alpha \sum_{f \in N_e} x_f s_f \right) \cdot x_e. & \text{def. $d_e$} \\
        & \geq (1 - \alpha \cdot s_e)  \cdot \left(h(2) +0.14\left(s_e + \alpha \sum_{f \in N_e} x_f s_f \right)\right), \nonumber
        \end{align}
     where the last inequality follows from a simple linear approximation of the convex function $h$, yielding $h(2-x)\geq h(2)+0.14x$ for all $x\in \mathbb{R}$.
\end{proof}

\Cref{lem:saturated} provides improved bounds on the probability of $e$ being matched whenever $s_e$ or $\sum_{f \in N_e} x_f s_f$ is large.
In the next section, we prove improved lower bounds on the probability of $e$ being matched in the complementary scenario, where $s_e$ and $\sum_{f\in N_e} x_f s_f$ are both small.

\subsection{Lower Bounding $\Pr[e \text{ \normalfont matched} \wedge R_1(e)]$} \label{sec: a1-lower-bound}

In this section, we prove \Cref{lem:unsaturated}, providing a lower bound for the probability that edge $e=(v_1v_2)$ is matched despite having a single incident realized edge $f=(v_2v_3)$ arriving before time $t_e$. 
We will prove this lemma by lower bounding the probability of $R_1(e)$ occurring, with the single realized edge $f \in R_1(e)$ not being matched due to some realized edge $g=(v_3v_4)\in N_f \setminus N_e$ arriving at time $t_g\leq t_f$ having no incident realized edges arriving earlier (i.e., event $R_0(g)$), resulting in $g$ being matched before $f$ arrives. See \Cref{fig:witness}.

\begin{figure}[h] 
    \centering
    \includegraphics[width=0.6\textwidth]{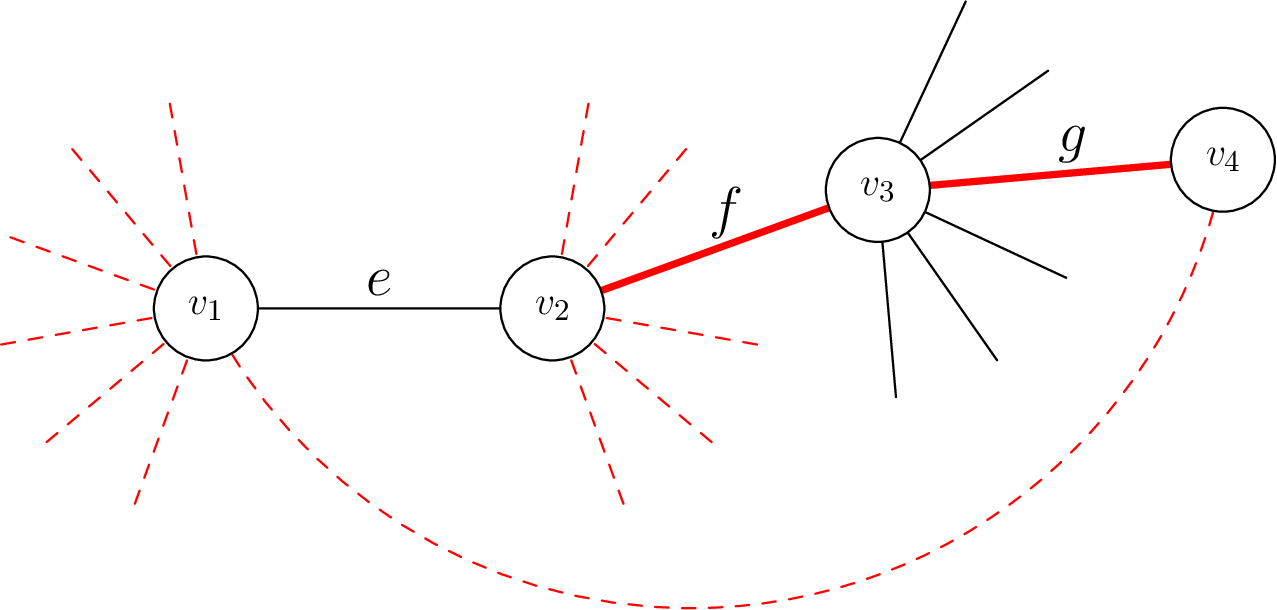}
    \caption{Our witness for $\mathds{1}[e\in \mu \wedge R_1(e)]$.}
    \subcaption*{Dashed edges are edges not realized before time $t_e$, while bold edges are realized before time $t_e$.}
    \vspace{0.2cm}
    \label{fig:witness}
\end{figure}

Fix any edge $g$ incident on $v_3$ and not incident on $e$, as in \Cref{fig:witness}. First, we note that conditioning on $Q(e) = \{f\}$ and specific arrival times (formalized below) does not decrease the probability of event $R_0(g)$, namely that zero edges incident on $g$ both arrive before $f$ and are realized.

\begin{restatable}{obs}{removeConditioning}
    \label{obs:removeConditioning}
    For any $0 \le c \le b \le a \le 1$, we have $$\Pr[ R_0(g) \mid Q(e) = \{ f \}, t_e = a, t_f = b, t_g = c ] \ge \Pr[ R_0(g) \mid t_g = c]\geq \e^{- (2 - 2x_g) y}.$$
\end{restatable}

The above observation removes the conditioning on events for edges other than $g$; it is useful when lower bounding the conditional probability that $g$ is matched and arrives before $f$.

\begin{lem}\label{lem:gmatched} For any $0 \le b \le a \le 1$ we have
     $$\Pr[ g \in \mu, t_g \le t_f \mid Q(e) = \{ f \}, t_e = a, t_f = b ] \ge x_g \cdot \frac{1}{2} (1 - 2\alpha) \cdot (1 - \e^{-2b}).$$
\end{lem}
\begin{proof}
Edge $g$ is matched if it is realized and $R_0(g)$ occurs. On the other hand, conditioned on $Q(e)=\{f\}, t_e=a,\, t_f=b, \, t_g=c$, the event that $g$ is realized is independent of $R_0(g)$, and has probability $x_g\cdot \e^{-x_g c}\cdot (1-\alpha\cdot s_g)$. Combined with \Cref{obs:removeConditioning}, this yields the following.
\begin{align*}
\Pr[ g \in \mu \mid Q(e) = \{ f \}, t_e = a, t_f = b, t_g = c ] & \ge x_g\cdot \e^{-x_g c} \cdot (1 - \alpha \cdot s_g) \cdot  \e^{-(2 - 2x_g) c}.
\end{align*}

Therefore, by taking total probability over $g$'s arrival time $c\leq b$, we see that
\begin{align*}
    \Pr[ g \in \mu \wedge t_g \le t_f \mid Q(e) = \{ f \}, t_e = a, t_f = b ] 
    &\geq x_g \cdot (1 - \alpha \cdot s_g) \cdot \int_0^b \e^{-(2 - 2x_g)c} \e^{-x_g c} \, dc \\
    &\ge x_g \cdot (1 - 2\alpha) \cdot \int_0^b \e^{-2c} \, dc \\
    &= x_g \cdot \frac{1}{2}(1 - 2\alpha) \cdot (1 - \e^{-2b}), 
\end{align*}
where the second inequality relied on $s_g\leq 2$ and $x_g\geq 0$ implying $2-x_g\le 2$.
\end{proof}

Next, fix an edge $f\in N_e$, and let $B_f$ denote the event that some edge $g \in N_f \setminus N_e$ with $t_g \le t_f$ is matched (hence \emph{blocking} edge $f$ from being matched).
Note that the single edge $g\in N_f\cap N_e$ (if any) must have $x_g\leq m_e$, by definition of $m_e = \max \{ x_f \mid f \in N_e, \text{ there exists a triangle containing } e, f \}.$
As the events $\{ g \in \mu\}$ are disjoint for distinct $g \in N_f \setminus N_e$, summing the individual bounds from \Cref{lem:gmatched} implies the following corollary. 

\begin{restatable}{cor}{PrBf}\label{cor:PrBf}
For any $0\le b \le a \le 1$ we have $$\Pr[ B_f \mid Q(e) = \{ f \}, t_e = a, t_f = b ] 
\ge (d_f - m_e - 1)^+ \cdot \frac{1}{2}(1 - 2\alpha) \cdot (1 - \e^{-2b}).$$
\end{restatable}

For our analysis, we are concerned with the event that $R_1(e)$ occurs, and that the unique realized edge $f$ is blocked. We denote this event by $U_e := \bigcup_{f\in N_e}[Q(e)=\{f\} \wedge B_f]$. The following lemma, whose proof is deferred to \Cref{app:omittedproofs3}, follows directly (after some calculation) from \Cref{cor:PrBf}.

\begin{restatable}{lem}{blockingfff}\label{lem: blocking-f}
For any $a \in [0,1]$, we have that $$ Pr[ U_e \mid t_e = a] \ge \frac{1}{2}  (1 - 2\alpha )^2  \sum_{f \in N_e} x_f (d_f - m_e - 1)^+ \cdot \e^{-a d_e + a x_f} \cdot \left( \frac{1 - \e^{-x_f a}}{x_f} - \frac{1 - \e^{-a(x_f+2)}}{x_f+2} \right).$$
\end{restatable}

We are now ready to prove our main lemma. 
\begin{proof}[Proof of \Cref{lem:unsaturated}]
Since $e$ is matched and $R_1(e)$ hold if $e$ is active and $U_e$ occurs (the latter two events being independent conditioned on the arrival time of $e$), we have

\begin{align*}
\Pr[e \in \mu \wedge R_1(e)] &\ge x_e \cdot \int_0^1 a_2(a, x_e, s_e) \cdot \Pr[U_e \mid t_e = a]  \, da.
\end{align*}
Let $z(x):=\int_0^1 \e^{-2a + ax_f}  \left( \frac{1 - \e^{-x a}}{x} - \frac{1 - \e^{-a(x+2)}}{x+2} \right) \, da$. Substituting in our bound from \Cref{lem: blocking-f} on $\Pr[U_e \mid t_e = a]$, and noting that $d_e + x_e \le 2$, we find, after some rearranging, that 
\begin{align*}
\Pr[e \in \mu \wedge R_1(e)]\geq  x_e \cdot \frac{1}{2}(1 - 2\alpha )^2 (1 - \alpha \cdot s_e)  \sum_{f \in N_e} x_f (d_f - m_e - 1)^+ \cdot z(x_f).
\end{align*}

We simplify this bound via the following numerical fact, whose proof is deferred to \Cref{app:omittedproofs3}.

\begin{restatable}{fact}{factBoundingIntofA}\label{fact:boundingIntofA}
    For any $x \in [0,1]$, we have that $x\cdot z(x) \ge 0.055x.$
\end{restatable}

Using this fact and substituting in $s_f = 2 - d_f - x_f$, we obtain the claimed bound.
\begin{align*}
    \Pr[e \in \mu \wedge R_1(e)] &\ge \left( (1 - \alpha \cdot s_e)  (1-  2\alpha)^2  \cdot \sum_{f \in N_e}  x_f (1 - m_e - x_f - s_f)^+ \cdot 0.0275  \right) \cdot x_e. \qedhere
\end{align*}
\end{proof}

\subsection{Reducing to a 5-variable problem}\label{sec:five-var}

We improve upon the bound of $h(2)=\frac{1}{2}(1-\e^{-2})$ by combining \Cref{lem:saturated} and \Cref{lem:unsaturated}. For ease of reference, we restate these two lemmas here.

\saturatedbound*
\unsaturatedbound*

\begin{proof}[Proof idea of \Cref{thm:ro-ocrs}]
The following is an easy lower bound on the probability of $e$ being matched.
\begin{align*}
    \Pr[e\in \mu] \geq \Pr[e\in \mu \wedge R_0(e)] + \Pr[e\in \mu \wedge R_1(e)] = (\star).
\end{align*}
Here, $(\star)$ is the lower bound obtained by summing the lower bounds of lemmas \ref{lem:saturated} and \ref{lem:unsaturated}. Before unpacking this (rather unwieldy) expression, we first simplify it as follows. 
First, we note that for $\alpha \ge 0.12$, the coefficient of $s_f$ in $(\star)$ is at least $x_e\cdot (1-\alpha\cdot s_e)\cdot x_f( 0.14 \alpha  - 0.0275 ( 1 - 2 \alpha)^2) \geq 0$.
Recalling that $s_f\geq x_f$, we have that $(\star)$ is made no larger by replacing each occurrence of $s_f$ with $x_f$, resulting in a lower bound $(\star)\geq (\dag)$. 

To lower bound $(\dag)$, we need a way to relate the $\sum_{f \in N_e} x_f^2$ to $m_e$. To do so, let $$f^* := \arg \max_{f \in N_e} \{ x_f : \text{ there exists a triangle containing } e, f \}.$$ So, $x_{f^*} = m_e$, with the possibility that $f^* = \emptyset$ if $m_e = 0$. Let $N_e^*$ denote $N_e \setminus \{ f^*\}$, and  define $d^{\text{big}}_e: = \sum_{f\in N_e^* : x_f \geq (1 - m_e)/2} x_f$ to be the contribution to $d_e$ due to edges with $x$-value at least $(1 - m_e)/2$. In $(\dag)$, the coefficients of $x_f^2$ for $f \in N_e^*$ with $x_f < (1-m_e)/2$ are given by $0.14 \alpha - 2 \cdot 0.0275 (1 - 2\alpha)^2$, which is positive for $\alpha \ge 0.171$. Using also that 
$d^{\text{big}}_e \cdot \frac{1 - m_e}{2} \le \sum_{f \in N_e: x_f \ge (1 - m_e)/2} x_f^2$ we have that $(\dag)$ is lower bounded by 
\begin{align*}
 (1 - \alpha \cdot s_e) \left( h (2)   +0.14 \left( s_e + \alpha m_e^2 + \alpha\cdot d^{\text{big}}_e \cdot \frac{1-m_e}{2} \right) +  0.0275 (1 - 2\alpha )^2 \cdot   (d_e - d^{\text{big}}_e) \cdot (1 - m_e)  \right) x_e.
\end{align*}

We conclude that the algorithm's balancedness for fixed $\alpha \in[0.171, 0.5]$ is at least the solution to the following five-variable optimization problem.
\begin{align*}
    \min & \enspace (1 - \alpha \cdot s_e) \Big( h (2)   +0.14 \left(s_e + \alpha m_e^2 + \alpha\cdot d^{\text{big}}_e \cdot \frac{1-m_e}{2} \right) +  0.0275 (1 - 2\alpha )^2 \cdot   (d_e - d^{\text{big}}_e) \cdot (1 - m_e)  \Big)  \\
    \textrm{s.t. } &   s_e + d_e = 2-x_e \\
    & d_e \geq d^{\text{big}}_e \\
    & d_e \leq 2(1-x_e) \\
    & m_e \leq 1 \\
    & x_e,s_e,d_e,d^{\text{big}}_e, m_e \geq 0.
\end{align*}

Numeric simulations with $\alpha=0.171$ yield a lower bound of $\corrgap$. Furthermore, in the bipartite case, we have only a four-variable optimization problem, as $m_e$ can be fixed at 0. In this case, numeric simulations yield a lower bound of $\bipcorrgap$.
\end{proof}

\section{Stochastic OCRS for Matching with Patience}\label{sec: patience}
In this section, we show that our approach can also be used to improve on the current state-of-the-art algorithm for the \emph{patience-constrained} setting frequently studied in the stochastic probing literature. In this extension of our problem, each vertex $u$ has an integer patience constraint $\ell_u$, specifying that at most $\ell_u$ edges incident on $u$ can be probed. By \Cref{lemma:reduction-patience}, we restrict our attention in this section to the stochastic OCRS model as in \Cref{def:stochastic-ocrs}. For the remainder of this section, we use the notation $x_e := y_e\cdot p_e$ (recall that $\vec{x}$ is in the matching polytope by \Cref{def:stochastic-ocrs}).

Our goal in this section is to prove the following theorem.

\patiencetheorem*

Let $B_u$ be the event that vertex $u$ is \emph{blocked}, either by an edge incident on $u$ being matched or $u$ running out of patience. A useful lemma for us, proven by Brubach et al.~\cite[Lemma 2]{brubach2021improved} for the generic \Cref{alg:stochastic-ocrs}, is that the worst-case bound for $\Pr[\bar{B}_u]$ is achieved when for all edges $e$ incident on $u$, we have $p_e \in \{0, 1\}$.

\begin{lem}[\cite{brubach2021improved}]\label{lem:integer-probs}
    For any vertex $u$, the value of $\Pr[\bar{B}_u]$ is made no larger by assuming that for all edges $e \ni u$, we have $p_e \in \{0, 1\}$.  
\end{lem}

The intuition behind this lemma is that for any edge $e \ni u$ with $p_e \in (0, 1)$, if we split it into two edges with probabilities 0 and 1, then the corresponding LP solution is still feasible and $\Pr[B_u]$ does not decrease. Let $E_0(u)$ and $E_1(u)$ denote the set of edges $e \ni u$ with $p_e = 0$ and $p_e = 1$, respectively.

\subsection{Analysis for General Graphs}
In the patience-constrained setting, we show that our new attenuation function $a_2(\cdot)$ can be used to improve on the state-of-the-art $0.382$-approximate algorithm of \cite{brubach2021improved}. We briefly note that if we use the attenuation function $a_1( \cdot )$ from \Cref{sec:warm-up-RO-edge-arrival} instead of $a_2( \cdot )$ in \Cref{alg:stochastic-ocrs}, we will get the same approximation ratio as in \cite[Section 3]{brubach2021improved} with a similar analysis. To analyze \Cref{alg:stochastic-ocrs} with patience constraints, we have the same high-level intuition as in \Cref{sec:going-beyond}: for each edge $e$ we will provide a lower bound for $\Pr[e \in \mu \wedge R_0(e)]$ and $\Pr[e \in \mu \wedge R_1(e)]$.

\subsubsection{Lower bounding $\Pr[e \in \mu \wedge R_0(e)]$}

First, we provide a lower bound for the event that no edges in $N_e$ arrived before $e$ and were realized, and $e$ was matched.

\begin{lem}\label{lem:a-0-patience}
For each edge $e\in E$, we have the following.
    $$\Pr[e \in \mu \wedge R_0(e)] \geq (1 - \alpha \cdot s_e)\left(0.382 + 0.117 \left(s_e + \alpha \cdot \sum_{f \in N_e} x_f s_f \right)\right)\cdotp x_e.$$
\end{lem}

\begin{proof}

By \Cref{lem:integer-probs}, we can assume that all edges incident on $e$ have integer success probabilities. There are two cases for edge $e = (uv)$ to be blocked: (i) some edge $f \in E_1(u) \cup E_1(v)$ arrives before $e$ and is realized, or (ii) $u$ or $v$ loses patience before $e$ is considered because of the probing of edges in $E_0(u) \cup E_0(v)$. (Note that any edge $E_1(u) \cup E_1(v)$ that we probe is successful, while edges in $E_0(u) \cup E_0(v)$ cannot be realized.) Define $\ell_u(e):=\{f\ni u \mid f\in E_0(u),\, t_f < t_e\}$. Then, note that if $\ell_u(e) \le l_u - 1$, which is independent of the event that $e$ is blocked by (i), it is not possible that $e$ is blocked by $u$ having lost patience despite being free (a similar statement holds for $v$). Putting it together, we have
    \begin{align*}
    &\Pr[R_0(e) \wedge \{ u, v \} \textrm{ have patience} \mid t_e = y] \\
    \geq &\Pr[R_0(e) \mid t_e = y]  \Pr[\ell_u(e) \leq \ell_u - 1] \Pr[\ell_v(e) \leq \ell_v - 1].
    \end{align*}
    
    As in the proof of \Cref{lem:saturated}, we have that 
    \begin{align*}
    \Pr[R_0(e) \mid t_e = y] \geq \e^{ -y( d_e - \alpha \cdot \sum_{f \in N_e} x_f s_f)}.
    \end{align*}
    Furthermore, by \cite[Lemma 11]{baveja2018improved}, $\Pr[\ell_u(e) \leq \ell_u - 1]$ is minimized when $\ell_u(e) \sim \text{ Pois}(y(\ell_u - 1))$. Let $\phi_{\ell_u}(y) := \Pr[\ell_u(e) \leq \ell_u - 1]$ when $\ell_u(e) \sim \text{ Pois}(y(\ell_u - 1))$. Thus
    \begin{align*}
        \Pr[R_0(e) \wedge \{u, v\} \textrm{ has patience} \mid t_e = y] \geq \e^{ -y( d_e - \alpha \cdot \sum_{f \in N_e} x_f s_f)} \cdot \phi_{\ell_u}(y) \cdot \phi_{\ell_v}(y).
    \end{align*}
    
    We finish the proof by taking total expectation over the time $e$ arrives:
    \begin{align}\label{eq: a0-patience-bound}
    \Pr[e \in \mu \wedge R_0(e)] & = x_e \cdot (1 - \alpha \cdot s_e)\cdot \int_0^1 \Pr[R_0(e) \wedge \{u, v\} \textrm{ have patience} \mid t_e = y] \cdot \e^{-x_e y} \, dy \nonumber
    \\
    & \geq x_e \cdot (1 - \alpha \cdot s_e)\cdot\int_0^1 \e^{ -y( d_e - \alpha \cdot \sum_{f \in N_e} x_f s_f)-x_e y} \cdot \phi_{\ell_u}(y) \cdot \phi_{\ell_v}(y) \, dy. 
    \end{align}
    
    Based on numerical simulations, we can show that the above is minimized when $\ell_u = 2$ and $\ell_v = 2$, yielding the following. \begin{align*}
         \Pr[e \in \mu \wedge R_0(e)] & \geq x_e\cdot(1 - \alpha \cdot s_e)\cdot\int_0^1 \e^{ -y( 2 - \alpha \cdot \sum_{f \in N_e} x_f s_f)} \cdot (\e^{-y} + y\e^{-y})^2 \cdot \, dy
         \\
         & = x_e\cdot(1 - \alpha \cdot s_e) \cdot \int_0^1 \e^{ -y( 4 - s_e - \alpha \cdot \sum_{f \in N_e} x_f s_f)} \cdot (1+y)^2 \, dy.
         \\
         &\geq (1 - \alpha \cdot s_e) \left(0.382 + 0.117 \left( s_e + \alpha \cdot \sum_{f \in N_e} x_f s_f \right) \right)\cdot x_e,
    \end{align*}
    where the last step is captured by \Cref{fact:factBoundingAPatiencea}, whose proof is deferred to \Cref{app:omittedproofs4}.
\end{proof}

Note that for $\alpha = 0$, \Cref{alg:main} captures the  \cite{brubach2021improved} result. By choosing $\alpha > 0$, we can guarantee that the bound in \Cref{lem:a-0-patience} is larger than 0.382 if $s_e$ or $\sum_{f \in N(e)}x_f s_f$ is large.

\subsubsection{Lower bounding $\Pr[e \in \mu \wedge R_1(e)]$}
In this section, we lower bound the probability that $e$ is matched despite a neighboring edge $f$ being realized and appearing before $e$, because $f$ is blocked.

\begin{lem}\label{lem:a1-lower-bound}
For each edge $e\in E$, we have the following.
    $$
    \Pr[e \in \mu \wedge R_1(e)] \ge \left( (1-  2\alpha)^2 (1 - \alpha \cdot s_e)  \cdot \sum_{f \in N_e}  x_f (1 - m_e - x_f - s_f)^+ \cdot 0.02  \right) \cdot x_e.
    $$
\end{lem}

Note that the difference between the analysis of this section and \Cref{sec: a1-lower-bound} is that when $g \in N_f$ which is incident on $v_3$ arrives before $f$, vertex $v_4$ must have remaining patience in order to argue that $f$ gets blocked by $g$. If $v_4$ has lost patience, $g$ will not block $f$, resulting in $e$ possibly not being free (even if conditioning on the fact that the only realized edge before $e$ is $f$). We prove a sequence of lemmas which help in the proof of \Cref{lem:a1-lower-bound}.

By \Cref{lem:integer-probs}, we can assume that all edges that are incident on $v_4$ excluding $g$, have integer success probabilities since it does not decrease $\Pr[B_{v_4}]$. Let $B_f$ be the probability that edge $f$ is blocked, either by an edge $g$ incident on $v_3$ or because $v_3$ ran out of patience. Then we have the following lemma and corollary that help to provide a lower bound for this event. 

\begin{restatable}{lem}{gmatchedpatience}\label{lem:g-matched-patience}
    For any $0 \le b \le a \le 1$ we have
     $$\Pr[ g \in \mu, t_g \le t_f \mid Q(e) = \{ f \}, t_e = a, t_f = b, \ell_{v_3}(f) \leq \ell_{v_3} - 1] \ge x_g \cdot \frac{(1- 2\alpha)\cdot (4 - (3b + 4)\e^{-3b})}{9}.$$
\end{restatable}

\begin{restatable}{cor}{PrBfpatience}\label{cor:PrBf-patience}
For any $0\le b \le a \le 1$ we have 
\begin{align*}
\Pr[B_f \mid Q(e) = \{ f \}, t_e = a, t_f = b ] 
&\ge (d_f - m_e - 1)^+ \cdot (1- 2\alpha) \cdot \frac{1}{9} \cdot (1 - (3b + 4)\e^{-3b}).
\end{align*}
\end{restatable}

We defer the proofs of the above lemma and corollary to \Cref{app:omittedproofs4}. Let $U_e$ denote the event that there is a unique edge $f$ in $Q(e)$, and furthermore that $B_f$ occurs ($U_e := \bigcup_{f\in N_e}[Q(e)=\{f\} \wedge B_f]$). The following lemma, whose proof is similar to the proof of \Cref{lem: blocking-f}, and is therefore deferred to \Cref{app:omittedproofs4},  lower bounds the probability of $U_e$ conditioned on the arrival time of edge $e$.
\begin{restatable}{lem}{somefblocked}\label{lem:some-f-blocked}
    For any $a \in [0,1]$,
    
    $$Pr[U_e \mid t_e = a] \geq \frac{1}{9}(1 - 2\alpha )^2 \sum_{f \in N_e} x_f (d_f - m_e - 1)^+\cdot \e^{-a d_e + ax_f} \cdot h_1(a, x_f),$$
    where $h_1(a,x_f) := \int_0^a \e^{-b x_f} \cdot  (4 - (3b + 4)\e^{-3b}) \, db$.
\end{restatable}

With this in place, we can complete the proof of \Cref{lem:a1-lower-bound}.

\begin{proof}[Proof of \Cref{lem:a1-lower-bound}]
Equipped with \cref{lem:some-f-blocked}, we are ready to provide a lower bound for $\Pr[e \in \mu \wedge R_1(e)]$ by taking total expectation over the time $e$ arrives:
\begin{align*}
    \Pr[e \in \mu \wedge R_1(e)]& \ge  x_e \cdot a_2(t_e,x_e,s_e) \\
    &\cdot \int_0^1 \Pr[U_e \mid t_e = a]\cdot \Pr[\ell_{v}(e) \leq \ell_{v} - 1 \mid t_e = a] \cdot \Pr[\ell_{u}(e) \leq \ell_{u} - 1 \mid t_e = a]  \, da,
\end{align*}
which is minimized when $\Pr[\ell_u(e) \leq \ell_u - 1]$ minimized when $\ell_u(e) \sim \text{ Pois}(y(\ell_u - 1))$ by \cite[Lemma 11]{baveja2018improved}. Assuming $\phi_{\ell_u}(y) := \Pr[\ell_u(e) \leq \ell_u - 1]$ when $\ell_u(e) \sim \text{ Pois}(y(\ell_u - 1))$, and $h_1(a,x_f) = \int_0^a \e^{-b x_f} \cdot  (4 - (3b + 4)\e^{-3b}) \, db$, we have that
\begin{align}\label{eq: a1-patience-bound}
\Pr[e \in \mu \wedge R_1(e)]& \ge \frac{1}{9}\cdot x_e \cdot(1 - 2\alpha )^2 (1 - \alpha \cdot s_e) \nonumber\\ 
&\cdot \int_0^1 \e^{-a x_e}  \sum_{f \in N_e} x_f (d_f - m_e - 1)^+\cdot \e^{-a d_e + ax_f} \cdot h_1(a,x_f) \cdot \phi_{\ell_u}(a)\cdot \phi_{\ell_v}(a)\, da,
\end{align}
which is minimized when $\ell_u = 2$ and $\ell_v = 2$. Then 
\begin{align*}
&\Pr[e \in \mu \wedge R_1(e)]\\
      &\ge \frac{1}{9}\cdot x_e(1 - 2\alpha )^2 (1 - \alpha\cdot s_e)  \sum_{f \in N_e} x_f (d_f - m_e - 1)^+ \int_0^1 \e^{-a x_e -a d_e + ax_f} \cdot h_1(a, x_f) (e^{-a} + ae^{-a})^2 \, da \\
      &\ge \frac{1}{9}\cdot x_e(1 - 2\alpha )^2 (1 - \alpha\cdot s_e)  \sum_{f \in N_e} x_f (d_f - m_e - 1)^+ \int_0^1 \e^{-4a + ax_f}\cdot h_1(a,x_f)\cdot (1+a)^2 \, da \\
      & \ge \frac{1}{9}\cdot x_e(1 - 2\alpha )^2 (1 - \alpha\cdot s_e)  \sum_{f \in N_e} x_f (d_f - m_e - 1)^+ \cdot 0.181.
\end{align*}
In the last inequality, we used \Cref{fact:patience-integral} to simplify the integral.
\begin{restatable}{fact}{patienceintegral}\label{fact:patience-integral}
    Let $h_1(a,x_f) = \int_0^a \e^{-b x_f} \cdot  (4 - (3b + 4)\e^{-3b}) \, db$. Then we have 
    $$\int_0^1 \e^{-4a + ax_f}\cdot h_1(a,x_f)\cdot (1+a)^2 \, da\geq 0.181.$$
\end{restatable}

Writing this in terms of $s_f = 2 - d_f - x_f$, we have 
\begin{align*}
    \Pr[e \in \mu \wedge R_1(e)] &\ge \left( (1-  2\alpha)^2 (1 - \alpha\cdot s_e)  \cdot \sum_{f \in N_e}  x_f (1 - m_e - x_f - s_f)^+ \cdot 0.02  \right) \cdot x_e. \qedhere
\end{align*}
\end{proof}

\subsection{Analysis for Bipartite Graphs with Patience Constraints on One Side}
In this section, we assume that the input graph is bipartite and we have patience constraints on only the vertices on one side (as motivated by applications in the gig economy). 

We show that with a small modification to the analysis for general graphs with patience, we can achieve a larger bound in the bipartite setting. We use \Cref{alg:stochastic-ocrs} and assume that patience number for vertices of one side is $\infty$. Note that all lemmas for general graphs also hold here. In order to get a better bound, we slightly change the analysis of \Cref{lem:a-0-patience} and \Cref{lem:a1-lower-bound} by assuming that only one endpoint of edge $e$ has a patience constraint. All the steps are similar to the analysis for general graphs with patience constraints and the only difference is the last steps in \Cref{lem:a-0-patience} and \Cref{lem:a1-lower-bound}. We defer the complete proofs of the following two lemmas to the appendix.

\begin{restatable}{lem}{bipartiteaa}\label{lem: bipartite-a-0}
For each edge $e\in E$, we have the following.

$$\Pr[e \in \mu \wedge R_0(e)] \geq (1 - \alpha \cdot s_e) \left(0.405 + 0.131 \left( s_e + \alpha \cdot \sum_{f \in N_e} x_f s_f \right) \right)\cdot x_e.$$
\end{restatable}

\begin{restatable}{lem}{bipartiteab}\label{lem: bipartite-a-1}
For each edge $e\in E$, we have the following.
    $$\Pr[e \in \mu \wedge R_1(e)] \geq  \left( (1-  2\alpha)^2 (1 - \alpha\cdot s_e)  \cdot \sum_{f \in N_e}  x_f (1 - x_f - s_f)^+ \cdot 0.023  \right) \cdot x_e.$$
\end{restatable}

\begin{proof}[Proof of \Cref{thm: patience-theorem}]
Combining \Cref{lem:a-0-patience} and \Cref{lem:a1-lower-bound}, for each edge $e$, we have the following lower bound for $\Pr[e \in \mu] / x_e$:
\begin{align*}
(1 - \alpha \cdot s_e)\cdot \left(0.382 + 0.117(s_e + \alpha \cdot \sum_{f \in N_e} x_f s_f) + (1-  2\alpha)^2  \cdot \sum_{f \in N_e}  x_f (1 - m_e - x_f - s_f)^+ \cdot 0.02\right),    
\end{align*}
for general graphs with patience on all vertices. Similar to \Cref{sec:five-var}, numeric simulations with $\alpha = 0.16$ yield a lower bound of $\patiencratio$.

When the input graph is a bipartite graph with patience constraint on one part, by combining \Cref{lem: bipartite-a-0} and \Cref{lem: bipartite-a-1}, for each edge $e$, we have the following lower bond for $\Pr[e \in \mu] / x_e$
\begin{align*}
(1 - \alpha \cdot s_e)\cdot \left(0.405 + 0.131(s_e + \alpha \cdot \sum_{f \in N_e} x_f s_f) + (1-  2\alpha)^2  \cdot \sum_{f \in N_e}  x_f (1 - x_f - s_f)^+ \cdot 0.023 \right), 
\end{align*}
for bipartite graphs with patience on one side. Numeric simulations with $\alpha = 0.162$ yield a lower bound of $\bippatienceratio$.
\end{proof}

\section{Discussion}

Our work developed a new RO-OCRS for matching and a new stochastic RO-OCRS for matching with patience constraints, and through this gave an approximation algorithm for the sequential pricing problem, with applications in the gig economy. We did so by designing a new attenuation function that takes into account how much contention an edge has for its endpoints. Our work leaves multiple interesting avenues for future research. 

A natural question is whether there is a more principled manner to pick $a(e)$ without changing the core of our analysis. We required multiple non-trivial lower bounds relying on an explicit form for $a(e)$ which seem hard to generalize to symbolic $a(e)$, but it is likely there is already room for improvement here. It would also be very interesting to extend our analysis by considering the probabilities of events $\Pr[e \in \mu \wedge R_k(e)]$ for $k \ge 2$.  

Our algorithm, as well as all prior work on (RO-)OCRS for matching constraints, works with the degree-bounded relaxation for matchings, which is integral for bipartite graphs, but not for general ones. 
A natural question is whether better (RO-)OCRS can be achieved for points in the matching polytope, i.e., points which also satisfy Edmonds' odd-set constraints \cite{edmonds1965maximum}. 
Can better bounds be achieved for this polytope, perhaps via a different algorithmic approach? 

Finally, we note that in the sequential pricing problem, we have complete control over the ordering in which we consider the edges; while RO-OCRS is a natural algorithmic primitive with many distinct applications, for our application the constraint that we consider edges in a uniformly random is self-imposed. Can a better approximation be achieved by considering edges in an order that is not uniformly random?

\paragraph{Acknowledgements.} This research was funded in part by NSF CCF1812919. We thank anonymous reviewers for their useful comments. 

\newpage 
\bibliographystyle{alpha}
\bibliography{abb,a}
\newpage 

\appendix

\section{Omitted Proofs of Section \ref{sec:framework}}\label{app:omittedproofs2}

\reductionpatience*
\begin{proof}

Define  $y_e := \sum_w y_{ew}$ and $p_e := \frac{\sum_w y_{ew} p_{ew}}{\sum_w y_{ew}}$; note that $(y_e)$ and $(p_e)$ form a valid input to a Stochastic OCRS algorithm by the constraints of \ref{LP-main}. 
Suppose that before running \Cref{alg:main}, we sample for each edge $e=(ij)$ and price $w$ independent Bernoulli variables $A_{ew} \sim \text{Ber}(p_{ew})$ \emph{in advance}, along with the weights $\{\pi_e\}$. (This clearly does not change the outcome of the algorithm.) Now, imagine that we run \Cref{alg:stochastic-ocrs} on the instance given by $\vec{y}, \vec{p}$ in the following way: we probe $e$ with probability $a(e)$ if and only if there exists some $w$ such that $\pi_e = w$, and we assume $e$ is active if and only if $A_{e \pi_e} = 1$. Note then when running \Cref{alg:stochastic-ocrs} in this way, we probe $e$ with probability $\left( \sum_w y_{ew} \right) \cdot a(e) = y_e \cdot a(e)$. Additionally, the probability $e$ is active conditioned on it being probed is precisely $\frac{\sum_w y_{ew} p_{ew}}{\sum_w y_{ew}} = p_e$. 

So, by the $c$-balancedness of \Cref{alg:stochastic-ocrs}, we know that $\Pr[e\in \mu] = c'\cdot y_e \cdot p_e$, for some $c'\geq c$.
Recall that $e$ is matched if and only if it is active (i.e., $A_{e \pi_e} = 1$), when processed it is free and has endpoints with patience, and an independent Bernoulli variable $\text{Ber}(a(e))$ comes up heads. Since these latter two events are independent of $e$'s active status, we have that $\Pr[e \in \mu \mid e \text{ active}] = c'$.
Indeed, since the joint event that $e$ is free and has endpoints with patience when processed, and $\text{Ber}(a(e))=1$  are both independent of $\pi_e$ and all the $A_{ew}$ variables, we have that for each $w$,  
$$\Pr[e \in \mu \mid \pi_e = w \wedge A_{e \pi_e} = 1 ] = \Pr[e\in \mu \mid A_{e \pi_e} = 1 ] = c'.$$
Consequently, the expected revenue satisfies the following:
\begin{align*}
    \E \left[ \sum_{j\in A} v_j - \sum_{e\in \mu} \pi_e \right] & = \sum_{ij}\sum_w (v_j - w)\cdot c'\cdot \Pr[\pi_e = w] \cdot \Pr[A_{e\pi_e} = 1 \mid \pi_e = w] \\
    & = \sum_{ij}\sum_w (v_j - w)\cdot c'\cdot y_{ew}\cdot p_{ew} \\
    & = c'\cdot OPT(\text{\textbf{LP-Pricing}}) \geq c\cdot OPT(\text{\textbf{LP-Pricing}}).
\end{align*}
Combining the above with \Cref{LP>=OPT}, the lemma follows.
\end{proof}

\section{Omitted Proofs of Section \ref{sec:going-beyond}}\label{app:omittedproofs3}

\removeConditioning*
\begin{proof}
    Since arrival times and realization status are independent, the LHS above is precisely
    $\Pr[ R_0(g) \mid Q(e) = \{ f \}, t_e = a, t_f = b, t_g = c ] = \prod_{e'\in N_g\setminus N_e} \Pr[e' \text{ not realized and arrived by time } c]$.
    The first inequality follows since the second term is precisely the same product taken over all edges in $N_g$, and is therefore no greater. Finally, the last inequality follows from our choice of attenuation function $a_2(g)\leq a_1(g)$, and therefore by the same argument as in the proof of \Cref{lem:analysisOfAlg2}, the probability that zero realized edges $g'\in N_g$ arrive before time $t_g$, is no smaller than under the attenuation function $a_1(\cdot)$, as calculated in \Cref{ineq:noblockingedge}.
\end{proof}

\PrBf*
\begin{proof}
    The events $E_g:= [g \in \mu, t_g\leq t_f]$ for $g\ni v_3$ are disjoint, since $v_3$ is matched at most once. Moreover, $B_f \supseteq \bigcup_{g\ni v_3, g \not \in N_e} E_g$. Therefore,
    \begin{align*}
    \Pr[ B_f \mid Q(e) = \{ f \}, t_e = a, t_f = b ] & \geq \sum_{g\ni v_3,\, g \not \in N_e} \Pr[ E_g \mid Q(e) = \{ f \}, t_e = a, t_f = b ] \\
    & \geq \sum_{g\ni v_3,\, g \not \in N_e} x_g \cdot \frac{1}{2} (1 - 2\alpha) \cdot (1 - \e^{-2b}) & \textrm{\Cref{lem:gmatched}} \\
    & \geq (d_f - m_e - 1)^+ \cdot \frac{1}{2}(1 - 2\alpha) \cdot (1 - \e^{-2b}),
    \end{align*}
    where the last inequality follows from the definition of $d_f=\sum_{g\in N_f} x_g$ together with the matching constraint $\sum_{g\ni v_3, g \neq f} x_g\leq 1$ and the definition of $m_e$.
\end{proof}

\blockingfff*
\begin{proof}
Fix an edge $f$. First, observe that the probability all edges $e' \in N_e \setminus \{f\}$ are either not realized or arrive after $t_e = a$ is
\begin{align*}
\Pr[Q(e)\subseteq \{f\} \mid t_e=a] & =  \prod_{e' \in N_e \setminus \{ f \}} \left( 1 - \int_{0}^a  x_{e'} \cdot a_2(b, x_{e'}, s_{e'}) \, db \right) \\
&\ge \prod_{e' \in N_e \setminus \{ f \}} \left( 1 - \int_{0}^a  x_{e'} \cdot \e^{-b x_{e'}} \, db \right) \\
&= \e^{- a \left( \sum_{e' \in N_e \setminus \{ f \} } x_{e'} \right) } \\
&\ge \e^{-a d_e + ax_f}.
\end{align*}
Also, the events $\{Q(e) = \{f\} \mid f\in N_e\}$ are disjoint, since $Q(e) = \{ f\}$ implies that $f$ is the \emph{unique} realized edge incident on $e$ arriving before $t_e$. Hence, using \Cref{cor:PrBf}, we have
\begin{align*}
\Pr[ U_e \mid t_e = a ]
&\geq  \sum_{ f \in N_e } \int_{0}^a x_f \cdot a_2(b, x_f, s_f) \cdot \e^{-a d_e + a x_f} \cdot \Pr[ B_f \mid Q(e) = \{ f\}, t_e = a, t_f = b ]  \, db \\
&\geq \sum_{ f \in N_e } \int_{0}^a x_f \cdot a_2(b, x_f, s_f) \cdot \e^{-a d_e + ax_f} \cdot (d_f - m_e - 1)^+  \cdot \frac{1}{2}(1 - 2 \alpha)(1 - \e^{-2b}) \, db   \\
&=  \frac{1}{2}(1 - 2 \alpha) (1 - \alpha \cdot s_f)   \sum_{f \in N_e} x_f (d_f - m_e - 1)^+ \cdot \e^{-a d_e + ax_f} \cdot \int_0^a \e^{-b x_f} \cdot  (1 - \e^{-2b}) \, db \\
 &\ge \frac{1}{2}(1 - 2\alpha )^2 \sum_{f \in N_e} x_f (d_f - m_e - 1)^+ \cdot \e^{-a d_e + ax_f} \cdot \left( \frac{1 - \e^{-x_f a}}{x_f} - \frac{1 - \e^{-a(x_f+2)}}{x_f+2} \right). \qedhere 
\end{align*}
\end{proof}

\factBoundingIntofA*
\begin{proof}
If $x= 0$ the bound clearly holds. Otherwise, let $x\in (0, 1]$. We first directly compute that 
\begin{align*}
\int_0^1 \e^{-2a + ax}   \left( \frac{1 - \e^{-x a}}{x} - \frac{1 - \e^{-a(x+2)}}{x+2} \right)\, da &=  \frac{2 \e^2 (x^2 - 4) - (x-2)x_f - \e^4 (x+2)x + 8\e^{x+2}}{4 \e^4 x (x^2 - 4)} \\
&= \left( \frac{2 \e^2(x^2 - 4) + 8\e^{x+2}}{4 \e^4 x (x^2 - 4)} \right) - \frac{1}{4 \e^4 (x+2)} - \frac{1}{4(x-2)}
\end{align*}

Clearly $-\frac{1}{4(x-2)} \ge \frac{1}{8}$ and $-\frac{1}{4 \e^4(x+2)} \ge -\frac{1}{8 \e^4}$ for $x \in (0,1]$. Additionally, we observe that $ \frac{2 \e^2(x^2 - 4) + 8\e^{x+2}}{4 \e^4 x (x^2 - 4)} $ is a decreasing function in $x$ for $x \in (0,1]$. Hence 
$$\frac{2 \e^2(x^2 - 4) + 8\e^{x+2}}{4 \e^4 x (x^2 - 4)} \ge \lim_{x \rightarrow 0} \frac{2 \e^2(x^2 - 4) + 8\e^{x+2}}{4 \e^4 x (x^2 - 4)} = -\frac{1}{2\e^2}.$$ Hence, for $x \in (0,1]$ we have 
\begin{align*}\int_0^1 \e^{-2a + ax}   \left( \frac{1 - \e^{-x a}}{x} - \frac{1 - \e^{-a(x+2)}}{x+2} \right) \, da \ge & \frac{1}{8} - \frac{1}{8 \e^4} - \frac{1}{2 \e^2} \ge 0.055. \qedhere
\end{align*}
\end{proof}

\section{Omitted Proofs of Section \ref{sec: patience}}\label{app:omittedproofs4}

\begin{restatable}{fact}{factBoundingAPatiencea}\label{fact:factBoundingAPatiencea}
    $\int_0^1 \e^{ -y( 4 - s_e - \alpha \cdot \sum_{f \in N_e} x_f s_f)} \cdot (1+y)^2 \, dy \geq 0.382 + 0.117 \left( s_e + \alpha \cdot \sum_{f \in N_e} x_f s_f \right)$.
    \end{restatable}
\begin{proof}
    Let $K = s_e + \alpha \cdot \sum_{f \in N_e} x_f s_f$. Since $s_e + \alpha \cdot \sum_{f \in N_e} x_f s_f \geq 0$, we need to show that for $K \geq 0$, 
    \begin{align*}
        \int_0^1 \e^{ -y( 4 - K)} \cdot (1+y)^2 \, dy \geq 0.382 + 0.117K.
    \end{align*}
    Let $g(K) = \int_0^1 \e^{ -y( 4 - K)} \cdot (1+y)^2 \, dy$. First, note that $g$ is increasing in $K$; hence, the minimum occurs at $K=0$,
    \begin{align*}
        g(0) = \int_0^1 \e^{ -y( 4 - 0)} \cdot (1+y)^2 \, dy \geq 0.382.
    \end{align*}
    Moreover, by taking derivative of $g$, we have 
    \begin{align*}
        g'(K) = \frac{(K^2 - 12K + 38) + e^{K-4}(4K^3 - 56K^2 + 266K - 430)}{(K-4)^4},
    \end{align*}
    which is an increasing function of $K$. Hence, $g'$ minimized at $K=0$ on $[0, \infty)$. Therefore, we have $g'(K) \geq g'(0) \geq 0.117$, which completes the proof. 
\end{proof}

\gmatchedpatience*
\begin{proof}
    Similar to the proof of the \Cref{lem:gmatched}, edge $g$ is matched if it is realized, $R_0(g)$ occurs, and $v_4$ has patience since we condition on $\ell_{v_3}(f) \leq \ell_{v_3} - 1$. Since, the event that $g$ is active and $R_0(g)$ are independent, we have that
    \begin{align*}
    & \Pr[ g \in \mu \mid Q(e) = \{ f \}, t_e = a, t_f = b, t_g = c, \ell_{v_3}(f) \leq \ell_{v_3} - 1] \\
    \geq & x_g \cdot a_2(c,x_g, s_g)\cdot \Pr[ R_0(g) \wedge \ell_{v_4}(g) \leq \ell_{v_4} - 1 \mid Q(e) = \{ f \}, t_e = a, t_f = b, t_g = c, \ell_{v_3}(f) \leq \ell_{v_3} - 1].
    \end{align*}
    Since arrival times and activation status are independent, we have that
    \begin{align*}
    & \Pr[ g \in \mu \mid Q(e) = \{ f \}, t_e = a, t_f = b, t_g = c, \ell_{v_3}(f) \leq \ell_{v_3} - 1] \\
    \ge & x_g \cdot a_2(c,x_g,s_g) \cdot \Pr[R_0(g) \wedge \ell_{v_4}(g) \leq \ell_{v_4} - 1 \mid t_g = c]   \\ \ge & x_g \cdot a_2(c,x_g,s_g) \cdot \Pr[R_0(g) \mid t_g = c]\cdot \Pr[\ell_{v_4}(g) \leq \ell_{v_4} - 1 \mid t_g = c] \\
    \geq & x_g\cdot \e^{-x_g c} \cdot (1 - \alpha \cdot s_g) \cdot  \e^{-(2 - 2x_g)c} \cdot \phi_{\ell_{v_4}}(c),
    \end{align*} 
    
    where $\phi_{\ell_{v_4}}(y) := \Pr[\ell_{v_4}(e) \leq \ell_{v_4} - 1]$ when $\ell_{v_4}(e) \sim \text{ Pois}(y(\ell_{v_4} - 1))$. By integrating over arrival time of $t_g \leq b$, we have that
    
    \begin{align*}
        \Pr[ g \in \mu, t_g \le t_f & \mid Q(e) = \{ f \}, t_e = a, t_f = b, \ell_{v_3}(f) \leq \ell_{v_3} - 1] \\
    \ge & x_g \cdot (1 - \alpha \cdot s_g) \cdot \int_0^b \e^{-(2 - 2x_g)c} \e^{-x_g c} \cdot \phi_{\ell_{v_4}}(c)\, dc \\
    \ge & x_g \cdot (1 - 2\alpha) \cdot \int_0^b \e^{-2c} \cdot \phi_{\ell_{v_4}}(c)\, dc.
    \end{align*}
    Based on numerical simulations, we can show that the above is minimized when $\ell_{v_4} = 2$, hence
    \begin{align*}
        \Pr[ g \in \mu, t_g \le t_f & \mid Q(e) = \{ f \}, t_e = a, t_f = b, \ell_{v_3}(f) \leq \ell_{v_3} - 1] \\
        \ge & x_g \cdot (1 - 2\alpha) \cdot \int_0^b \e^{-3c}(1 + c)\, dc \\
        = & x_g \cdot (1- 2\alpha) \cdot \frac{1}{9} \cdot (4 - (3b + 4)\e^{-3b}).\qedhere
    \end{align*}
\end{proof}

\PrBfpatience*
\begin{proof}
Note that the probability that $f$ is blocked condition on $\ell_3(f) \geq \ell_3$ is one. Moreover, conditional events $E_g:= \mathds{1}[g \in \mu, t_g\leq t_f]$ over all $g\ni v_3$ are disjoint. Thus, we have
\begin{align*}
\Pr[B_f \mid Q(e) = \{ f \}, t_e = a, t_f = b ] &\geq \Pr[B_f \mid Q(e) = \{ f \}, t_e = a, t_f = b, \ell_{v_3}(f) \leq \ell_{v_3} - 1] \\
&\ge \sum_{g\ni v_3,\, g \not \notin N_e}\Pr[ E_g\mid Q(e) = \{ f \}, t_e = a, t_f = b, \ell_{v_3}(f) \leq \ell_{v_3} - 1]\\
& \ge (d_f - m_e - 1)^+ \cdot (1- 2\alpha) \cdot \frac{1}{9} \cdot (1 - (3b + 4)\e^{-3b}),
\end{align*}
where the last inequality is followed by \Cref{lem:g-matched-patience}.
\end{proof}

\somefblocked*

\begin{proof}
    Similar to the proof of \Cref{lem: blocking-f}, the probability that for each $e' \in N_e \setminus \{f\}$ either is not realized or arrives after $t_e = a$ is at least $\e^{-a d_e + ax_f}$. Moreover, since events $\{Q(e) = f\}_f$ are disjoint, then $\Pr[ \exists f \textup{ such that } Q(e) = \{ f \} \wedge B_f \mid t_e = a]$ is at least
    \begin{align*}
 & \sum_{ f \in N_e } \left( \int_{0}^a x_f \cdot a_2(t_f,x_f,s_f)\cdot \e^{-a d_e + a x_f} \cdot \Pr[ B_f \mid Q(e) = \{ f\}, t_e = a, t_f = b ]  \, db  \right), 
\end{align*}
which by \Cref{cor:PrBf-patience} is at least,
\begin{align*}
 &\sum_{ f \in N_e } \int_{0}^a x_f \cdot a_2(t_f,x_f,s_f)\cdot \e^{-a d_e + ax_f} \cdot (d_f - m_e - 1)^+ \cdot (1 - 2\alpha)\cdot \frac{1}{9}\cdot(4 - (3b + 4)\e^{-3b}) \, db  \\
 &=   \frac{1}{9}\cdot(1 - 2\alpha) (1 - \alpha \cdot s_f)   \sum_{f \in N_e} x_f (d_f - m_e - 1)^+\cdot \e^{-a d_e + ax_f} \cdot \left( \int_0^a \e^{-b x_f} \cdot  (4 - (3b + 4)\e^{-3b}) \, db \right) \\
  &\ge \frac{1}{9}\cdot(1 - 2\alpha )^2 \sum_{f \in N_e} x_f (d_f - m_e - 1)^+\cdot \e^{-a d_e + ax_f} \cdot h_1(a,x_f). \qedhere
\end{align*}
\end{proof}

\begin{fact}\label{fact:factBoundingAPatiencec}
    $\int_0^1 \e^{ -y( 3 - s_e - \alpha \cdot \sum_{f \in N_e} x_f s_f)} \cdot (1+y) \, dy \geq 0.405 + 0.131(s_e + \alpha \cdot \sum_{f \in N_e} x_f s_f).$
\end{fact}
\begin{proof}
    All steps of the proof are the same as in the proof of \Cref{fact:factBoundingAPatiencea}. Let $K = s_e + \alpha \cdot \sum_{f \in N_e} x_f s_f$ and $g(K) = \int_0^1 \e^{ -y( 3 - K)} \cdot (1+y) \, dy$. We note that $g(\cdot)$ is an increasing function of $K$, so the minimum occurs at $K=0$, where
    \begin{align*}
        g(0) = \int_0^1 \e^{ -y( 3 - 0)} \cdot (1+y) \, dy \geq 0.405.
    \end{align*}
    Also, we have that 
    \begin{align*}
        g'(K) = \frac{(K-5) + e^{K-3}(2K^2 - 15K + 29)}{(K-e)^3}.
    \end{align*}
    Note that $g'$ is an increasing function of $K$ that is minimized at $K=0$ on $[0, \infty)$. Therefore, we have $g'(K) \geq g'(0) \geq 0.131$, which completes the proof. 
\end{proof}

\patienceintegral*
\begin{proof}[Proof sketch]
This fact is verified computationally by direct calculation of the integral. 
\end{proof}

\bipartiteaa*
\begin{proof}
    With the exact same analysis as in \Cref{lem:a-0-patience}, we have \Cref{eq: a0-patience-bound}. Now without loss of generality, let $v$ be the vertex such that $\ell_v = \infty$. Hence, we have $\phi_{\ell_v}(y) = 1$. Therefore,
    \begin{align*}
        \Pr[e \in \mu \wedge R_0(e)] & \geq x_e\int_0^1 \e^{ -y( d_e - \alpha \cdot \sum_{f \in N_e} x_f s_f)} \cdot \phi_{\ell_u}(y) \cdot \e^{-x_e y}\cdot (1 - \alpha \cdot s_e) \, dy.
    \end{align*}
    Based on numerical simulations, we can show that the above is minimized when $t_u = 2$, yielding the following
    \begin{align*}
        \Pr[e \in \mu \wedge R_0(e)] & \geq x_e \cdot (1 - \alpha \cdot s_e)\cdot \int_0^1 \e^{ -y( d_e - \alpha \cdot \sum_{f \in N_e} x_f s_f)} \cdot (\e^{-y} + y\e^{-y}) \cdot \e^{-x_e y} \, dy
         \\
         & = x_e\cdot(1 - \alpha \cdot s_e)\cdot \int_0^1 \e^{ -y( 3 - s_e - \alpha \cdot \sum_{f \in N_e} x_f s_f)} \cdot (1+y) \, dy,
    \end{align*}
    which can be simplified by \Cref{fact:factBoundingAPatiencec}. Therefore, we have that 
    \begin{align*}
    \Pr[e \in \mu \wedge R_0(e)] &\geq (1 - \alpha \cdot s_e) \left(0.405 + 0.131 \left( s_e + \alpha \cdot \sum_{f \in N_e} x_f s_f \right) \right) \cdot x_e. \qedhere
    \end{align*}
\end{proof}

\begin{fact}\label{fact:patience-integral2}
    Let $h_1(a,x_f) = \int_0^a \e^{-b x_f} \cdot  (4 - (3b + 4)\e^{-3b}) \, db$. Then we have
    $$\int_0^1 \e^{-3a + ax_f}\cdot h_1(a,x_f)\cdot (1+a) \, da\geq 0.209.$$
\end{fact}
\begin{proof}[Proof sketch]
This fact is verified computationally by direct calculation of the integral. 
\end{proof}

\bipartiteab*
\begin{proof}
With exact same analysis of \Cref{lem:a1-lower-bound}, we have \Cref{eq: a1-patience-bound}. Without a loss of generality, let $v$ be the vertex such that $\ell_v=\infty$. Thus $\phi_{\ell_v}(a) = 1$, yielding the following
\begin{align*}
     \Pr[e \in \mu \wedge R_1(e)] & \ge x_e\cdot \frac{1}{9}\cdot(1 - 2\alpha )^2 (1 - \alpha \cdot s_e)\\ &\cdot \int_0^1 \e^{-a x_e}  \sum_{f \in N_e} x_f (d_f - 1)^+\cdot \e^{-a d_e + ax_f} \cdot h_1(a,x_f) \cdot \phi_{\ell_u}(a)\, da,
\end{align*}
 where $h_1(a,x_f) = \int_0^a \e^{-b x_f} \cdot  (4 - (3b + 4)\e^{-3b}) \, db$. Based on numerical simulations, we can show that the above is minimized when $t_u = 2$, yielding the following
\begin{align*}
        & \Pr[e \in \mu \wedge R_1(e)] \\
      &\ge x_e \cdot \frac{1}{9}(1 - 2\alpha )^2 (1 - \alpha \cdot s_e)  \sum_{f \in N_e} x_f (d_f - 1)^+ \int_0^1 \e^{-a x_e}  \e^{-a d_e + ax_f} \cdot h_1(a, x_f) \cdot (e^{-a} + ae^{-a}) \, da \\
      &\ge x_e \cdot \frac{1}{9}(1 - 2\alpha )^2 (1 - \alpha \cdot s_e)  \sum_{f \in N_e} x_f (d_f - 1)^+ \int_0^1 \e^{-3a + ax_f}\cdot h_1(a,x_f)\cdot (1+a) \, da \\
      & \ge x_e \cdot \frac{1}{9}(1 - 2\alpha )^2 (1 - \alpha \cdot s_e)  \sum_{f \in N_e} x_f (d_f - 1)^+ \cdot 0.209.
\end{align*}
In the last inequality, we used \Cref{fact:patience-integral2} to simplify the integral. Writing this in terms of $s_f = 2 - d_f - x_f$, we have 
\begin{align*}
    \Pr[e \in \mu \wedge R_1(e)] &\ge \left( (1-  2\alpha)^2 (1 - \alpha \cdot s_e)  \cdot \sum_{f \in N_e}  x_f (1 - x_f - s_f)^+ \cdot 0.023  \right) \cdot x_e. \qedhere
\end{align*}    
\end{proof}

\section{Ruling Out Natural Algorithms} \label{app:rulingout}

\subsection{Greedy According to Decreasing Weights}
If there is only one weight for each edge, then by sorting the edges based on their weights and querying in that order, we can achieve a half approximation. However, this greedy algorithm achieves \textbf{no constant approximation} if we have multiple weights for one edge. Consider a single edge with two different weights, one and two, with probabilities one and $\eps$, respectively. The greedy algorithm will gets a matching with expected weight of $2\eps$, however, the optimal algorithm solution has expected size of one.

\subsection{Greedy According to Decreasing Expected Weights}
One can modify this greedy algorithm in the previous subsection to query edges based on decreasing order of $w\cdot p_{ew}$. But this modified version achieves \textbf{no constant approximation}; Consider a star with edges $e_0, e_1, \ldots, e_k$ with a single weight per edge such that $w_{e_0} = 1 + \eps$ and $w_{e_i} = N^i$ for $ 1 \leq i \leq k$. Also, we have that $p_{e_i} =1/ N^i$. The modified greedy algorithm will get a matching with expected size of $1+\eps$ (picks $e_0$) but an algorithm that asks edges in the order of $(e_k, e_{k-1}, \ldots, e_1, e_0)$ will get a matching with expected size of roughly equal to $k$. Since we can make $k$ as large as we want, then the modified greedy is not a constant approximation.

\subsection{Trivial Attenuation is $1/3$-balanced}
In this section we consider the simplest instantiation of \Cref{alg:main}: the one with no attenuation, i.e., with $a(e)=1$ for all $e\in E$. Bruggmann and Zenklusen \cite{bruggmann2020optimal} show that this algorithm is $\frac{1}{3}$-balanced. The following simpler analysis recreates the same bound.
\begin{lem}\label{lem:warmup}
    Denote by $\Pr[e \textup{ free}]$ the probability that $e=(i,j)$ is free when inspected. Then, 
    $$\Pr[e \textup{ free}] \geq \frac{1}{3}.$$
\end{lem}
\begin{proof}
    For this instantiation of \Cref{alg:main}, each edge $e$ is matched iff it is free (at time $t_e)$. And indeed, we have the following.
    \begin{align*}
        \Pr[e \textup{ free}] & \geq \Pr[R_0(e)]
         = \int_{0}^1 \prod_{f\in N_e}(1-z\cdot x_f) \,\,dz \geq \int_0^1 (1-z\cdot (1-x_e))^2\,\,dz 
        \geq \int_0^1 (1-z)^2\,\,dz = \frac{1}{3},
    \end{align*}
    where the second inequality relied on \Cref{fact:warmup} below, combined with the matching constraint $\sum_{f\ni v}x_f\leq 1-x_e$ for both vertices $v\in e$. The third inequality follows from non-negativity of $x_e.$
\end{proof}

The following fact is easily proven by induction.
\begin{fact}\label{fact:warmup}
    Let $R_1,R_2,\dots,R_n\geq 0$. Then $\prod_{i=1}^n (1-R_i) \geq 1-\sum_{i=1}^n R_i$. 
\end{fact}
The constant $1/3$ is tight for the algorithm without attenuation (see \cite[Figure 6]{bruggmann2020optimal}).

\section{Reduction to Single Weight without Patience with a $\nicefrac{1}{2}$-Factor Loss}
In this section, we show that after solving \ref{LP-main} (with $\ell_v = \infty$ for all $v$), we can incur a multiplicative loss of no worse than $1/2$ and reduce the problem to the case where for each edge $e$, there is only one weight $w$ where $y_{ew} > 0$.

We first prove that there exists an optimal solution to our LP where each edge's $y$-values are supported on at most two weights. 

\begin{lem}
    There exists an optimal solution $\vec{y}'$ to our LP where for any edge $e$, there are at most two different weights $w_1$ and $w_2$ such that $y'_{ew_1} > 0$ and $y'_{ew_2} > 0$.
\end{lem}

\begin{proof}
    Let $\vec{y}^*$ denote an optimal solution to our LP, and recall our notation $x_{e} = \sum_w y^*_{ew}\cdot p_{ew}$. For a fixed edge $e$, consider the following LP for variables $\{ y_{ew} \}_w$:
    
\begin{align*}
    \max & \sum_{w} y_{ew}\cdot p_{ew}\cdot (v_j - w) \\
    \textrm{s.t.} \hspace{1em} & \sum_{w} y_{ew}  \leq 1 \\
    & \sum_{w} y_{ew}\cdot p_{ew}  \leq x_e \\
    & y_{ew} \geq 0.
\end{align*}

Any linear program in standard form with maximization objective function has an optimal solution at an extreme point. Since we have only two non-trivial constraints in the above LP, any extreme point has at most two non-zero values. Let $\vec{y}'$ denote this solution; note furthermore that replacing $y^*_{ew}$ with $y'_{ew}$ for all $w$ in our original LP results in a solution that is still feasible and has no loss in objective. \end{proof}

Now for each edge $e$, let $W(e)$ denote the weight $w$ that maximizes $y_{ew}\cdot p_{ew} \cdot (v_j - w)$. Since there exists an optimal solution $\vec{y}'$ with only two non-zero weights for each edge, we have that

\begin{align*}
     y_{eW(e)}\cdot p_{eW(e)} \cdot W(e) \geq \frac{1}{2}\cdot \sum_{w} y_{ew}\cdot p_{ew} \cdot (v_j - w)
\end{align*}

Therefore, any $\alpha$-approximate algorithm for LP solutions with a single weight on each edge, results in an $\alpha/2$-approximation algorithm for our original problem.

\section{Random Vertex Arrival}\label{sec:ro-vertex}
In the random vertex arrival setting, vertices of one side of the bipartite graph are present at the beginning. Vertices of the other side arrive one by one in random order with their incident edges. In this setting, after a vertex arrives, an online algorithm can probe incident edges to build a matching. Note that the algorithm is only allowed to probe feasible edges that do not violate matching constraints.

Similar to \Cref{sec:warm-up-RO-edge-arrival}, we let $\{ x_e \}_{e \in E}$ denote the realization probabilities in the random vertex arrival. Note that in this setting, we cannot sort edges based on a random time $t_e \sim \text{Unif}[0,1]$ for each edge because edges that are incident on a specific online vertex arrive together. Therefore, we use a random time $t_u \sim \text{Unif}[0,1]$ for each online vertex $u$ and we sort edges based on the tuple $(t_u, t_e)$ where $u$ is the online endpoint of edge $e$.

\begin{algorithm}
\SetAlgoLined
        \textbf{for} each edge $e$, sample a random time $t_e, t_u \sim \text{Unif}[0,1]$ for every edge $e$ and every online vertex $u$ 
        
        \For{each edge $e$, in increasing order of $(t_u, t_e)$ where $u$ is its online endpoint}{
             \If{$e$ is active and is free}{
    		        \textbf{with probability} $s(x_e, t_e, t_u)$, match $e$.
	        }
	  }
    \caption{}
    \label{alg:vanilla-vertex-arrival}
\end{algorithm}

We use the same attenuation function $e^{-x_e t_e}$ as \Cref{sec:warm-up-RO-edge-arrival}. Note that this attenuation function is not dependent on $t_u$.

\begin{lem}
    When running \Cref{alg:vanilla-vertex-arrival} with $s(x_e, t_e, t_u) = \e^{-t_e x_e}$, the probability we sell along edge $e$ is at least $ (1-1/\e)^2 \ge 0.399 \cdot x_e$. 
\end{lem}
\begin{proof}
Let $N_u$ be the set of edges incident on vertex $u$. Let $e = (u,v)$ where $u$ is online endpoint and $v$ is the offline endpoint. There are two types of edges that can block $e$: 1) $f \in N_u \setminus \{e\}$ such that $t_f < t_f$. 2) $g = (w, v) \in N_v \setminus \{e\}$ such that $t_w < t_u$.  Hence, if we condition on edge $e = (u, v)$ with $t_e=y_1$ and online endpoint arrival time $t_u = y_2$, we have the following inequality

\begin{align}\label{eq: cond-free-vertex}
\Pr[e \textrm{ free} \mid t_e = y_1, t_u = y_2]  \geq  \prod_{f\in N_u} \left( 1 - \int_0^{y_1} x_f \cdot \e^{- x_f z} \, dz \right) \prod_{g\in N_v} \left( 1 - y_2\cdot\int_0^{1} x_g \cdot \e^{- x_g z} \, dz \right).
\end{align}

The reason is that if an $f = (u, v')$ blocks $e$, it must arrive before $t_e$, edge $f$ must survive the attenuation step, and $f$ must be successful after the proposal. Moreover, for an edge $g = (u', v)$ to block $e$, online endpoint $u'$ must arrive before $t_u$ (note that $t_g$ does not play a role on whether $g$ arrives before $e$), it must survive the attenuation step, and $g$ must be succseful after the proposal. \Cref{eq: cond-free-vertex} follows by the fact that all these events are independent. Then we have that
\begin{align*}
    \Pr[e \textrm{ free} \mid t_e = y_1, t_u = y_2] & \geq  \prod_{f\in N_u} \left( 1 - \int_0^{y_1} x_f \cdot \e^{- x_f z} \, dz \right) \prod_{g\in N_v} \left( 1 - y_2\cdot\int_0^{1} x_g \cdot \e^{- x_g z} \, dz \right) \\
    & = \prod_{f\in N_u} \e^{-x_f y_1} \prod_{g\in N_v} \left( 1 - y_2\cdot(1 - \e^{-x_g}) \right)
\end{align*}
By \Cref{fact:concavity}, we have $\left( 1 - y_2\cdot(1 - \e^{-x_g}) \right) \geq \e^{-x_g y_2}$. Therefore we have the following.
\begin{align*}
    \Pr[e \textrm{ free} \mid t_e = y_1, t_u = y_2] & \geq \prod_{f\in N_u} \e^{-x_f y_1} \prod_{g\in N_v} \left( 1 - y_2\cdot(1 - \e^{-x_g}) \right)\\
    & \geq \prod_{f\in N_u} \e^{-x_f y_1} \prod_{g\in N_v} \e^{-x_f y_1} \\
    & = \e^{-y_1(1 - x_e)} \cdot \e^{-y_2(1 - x_e)} = \e^{-(1 - x_e)(y_1 + y_2)}.
\end{align*}
Hence, we have that
\begin{align*}
    \Pr[e \textrm{ successful}]  &= \Pr[e \text{ free at } (t_u, t_e) \text{ and } \text{Ber}(s(x_e, t_e, t_u)) = 1] \cdot x_e \\
    & \geq \left( \int_0^1 \int_0^1  \e^{-(1-x_e)(y_1 + y_2)} \cdot \e^{-y_1 x_e} dy_1 dy_2\right)\cdot x_e \\
    & \geq \left( \int_0^1 \int_0^1 \e^{-(y_1 + y_2)} dy_1 dy_2\right)\cdot x_e \\
    & \geq \left( 1-\frac{1}{\e} \right)^2 \cdot x_e \geq 0.399\cdot x_e. \qedhere
\end{align*}

\end{proof}

\end{document}